\documentclass[lettersize,journal]{IEEEtran}
\usepackage{amsmath,amsfonts}
\usepackage{array}
\usepackage[caption=false,font=normalsize,labelfont=sf,textfont=sf]{subfig}
\usepackage{textcomp}
\usepackage{stfloats}
\usepackage{url}
\usepackage{verbatim}
\usepackage{graphicx}
\usepackage{cite}
\usepackage{orcidlink}
\usepackage{enumitem}
\usepackage{algorithm}
\usepackage{algpseudocode}  
\usepackage{caption}
\usepackage{subfig}
\usepackage{amssymb}    
\usepackage{amsthm}     
\usepackage{hyperref}  
\usepackage[T1]{fontenc}

\usepackage{booktabs}  
\usepackage{bbm}        

\algrenewcommand\algorithmicrequire{\textbf{Input:}}
\algrenewcommand\algorithmicensure{\textbf{Output:}}

\def\BibTeX{{\rm B\kern-.05em{\sc i\kern-.025em b}\kern-.08em
		T\kern-.1667em\lower.7ex\hbox{E}\kern-.125emX}}
\usepackage{balance}
\makeatletter
\newcommand*{\rom}[1]{\expandafter\@slowromancap\romannumeral #1@}
\makeatother
\providecommand{\triangleq}{\stackrel{\triangle}{=}}

\newtheorem{theorem}{Theorem}

\newtheorem{proposition}{Proposition}

\theoremstyle{remark}

\begin{document}
	\title{Adaptive Distributed Physical-Layer Authentication and Attack Detection in 6G Non-Terrestrial Networks via Causal Meta-Learning}
	\author{Parsa~Rajabi~\orcidlink{0009-0005-4645-9444},
		Mohammad~Reza~Abedi~\orcidlink{0000-0003-4114-1339}~\IEEEmembership{Student~Member,~IEEE},
		Nader~Mokari~\orcidlink{0000-0001-5364-8888}~\IEEEmembership{Senior~Member,~IEEE},
		Paeiz~Azmi~\orcidlink{0000-0001-9736-3462}~\IEEEmembership{Senior~Member,~IEEE},
		Halim~Yanikomeroglu~\orcidlink{0000-0003-4776-9354}~\IEEEmembership{Fellow,~IEEE}
		\thanks{This work has been submitted to the IEEE for possible publication.
			Copyright may be transferred without notice, after which this version may
			no longer be accessible.}
		\thanks{P. Rajabi, MR. Abedi, N. Mokari, and P. Azmi are with the Department of Electrical and Computer Engineering, Tarbiat Modares University, Tehran, 14115-111, Iran (e-mail: parsa\_rajabi@modares.ac.ir; Mohammadreza\_abedi@modares.ac.ir; nader.mokari@modares.ac.ir; pazmi@modares.ac.ir).}
		\thanks{H. Yanikomeroglu is with the Department of Systems and Computer Engineering, Carleton University, Ottawa, K1S 5B6, Canada (e-mail: halim@sce.carleton.ca).}
	}

	\maketitle
	
\begin{abstract}
	Physical-layer authentication (PLA) in non-terrestrial networks (NTNs) is challenged by severe Doppler shifts, long delays, and fast channel variations, which cause distribution shifts and degrade conventional learning methods. Existing PLA schemes often rely on single features or generalize poorly to unseen environments. This paper proposes a secure adaptive framework for authentication in multi-zone networks (SAFA-MZ), a causal meta-learning framework for distributed PLA (DPLA) in NTNs. First, we design a multi-feature fingerprint that combines spatial, angular, combiner, subspace, and Doppler-delay features. The fingerprint is adaptive and distributed, as it fuses heterogeneous physical-layer features and measurements from multiple aerial nodes. Second, we formulate a structural causal model (SCM) to capture the relations among design choices, environmental factors, extracted features, and authentication outcomes. Third, we develop a model-agnostic meta-learning (MAML) strategy with invariant risk minimization (IRM) and causal consistency regularization for fast adaptation to unseen NTN environments with few labeled samples. Fourth, we propose a two-stage authentication scheme that performs local recognition and activates time-difference-of-arrival (TDOA) localization with a graph attention (GAT) network only when needed, which reduces backhaul overhead. Simulations show that SAFA-MZ achieves $92\%$ accuracy and $96\%$ AUC, outperforming centralized deep learning and single-feature baselines across diverse environments.
\end{abstract}
	

\section{Introduction}
\label{sec:Introduction}

\IEEEPARstart{P}{hysical}-layer authentication (PLA) verifies transmitter identity using channel characteristics and hardware impairments. Machine learning methods have been widely used for PLA \cite{ref53}, including channel prediction-based, classification-based, and similarity-based schemes \cite{ref46}. Compared with threshold-based hypothesis testing \cite{ref82, ref79, Chang2025PTPLA, ref88}, artificial intelligence (AI)-based authentication avoids fixed thresholds and improves robustness in dynamic environments \cite{ref81, ref51}.

Existing PLA methods rely on single features, tag-embedding \cite{Chang2025PTPLA}, or hybrid schemes \cite{HybridPLA}. In contrast, this paper uses multiple features for dynamic non-terrestrial networks (NTNs). Security in NTN is challenged by time-varying channels and hardware impairments \cite{Demir2025e}, which degrade authentication.
To address spoofing in such settings, we propose a multi-feature distributed PLA (DPLA) framework.
DPLA architectures, including semi-distributed, decision-level fusion, and fully distributed designs, offer different tradeoffs among accuracy, overhead, and robustness \cite{Zhang2024DPLA_overview}. Prior works have studied DPLA for millimeter-wave multiple-input multiple-output (mmWave MIMO) systems using beam-pattern deviations \cite{Zhang2025DistributedPLA}, dynamic weighted voting in smart grids \cite{Li2025DistributedPLA}, and missed-detection bounds for ultra-reliable low-latency communications (URLLC) \cite{Forssell2022}.

Deep learning-based fingerprint identification can be vulnerable to adversarial perturbations \cite{RFFIAdversarial} and cross-receiver distribution shifts caused by heterogeneous receiver hardware \cite{SCRFFI_paper}.
To improve robustness under distribution shifts and adversarial variations, we combine causal learning and meta-learning.
Causal learning extracts invariant relations and reduces spurious correlations, while meta-learning enables fast adaptation to unseen environments with few labels.
Related studies have used structural causal models (SCMs) for radio-frequency (RF) fingerprint purification \cite{Tang2025}, invariant graph learning \cite{Job2025}, counterfactual fairness \cite{CausalFairness2022}, reinforcement learning (RL) \cite{Roy2025, CausalMAPO, serra2026causal}, 6G explainability \cite{Arana2025}, wireless sustainability analysis \cite{Mata2025}, and semantic communication \cite{Thomas2025}. Meta-learning has also shown effectiveness in internet of things (IoT)-based multi-user PLA \cite{Zhao2026TMML} and low-Earth orbit (LEO) satellites \cite{DMMAFL}, and sparse fingerprint-based localization using graph-based spatial aggregation \cite{AGML_paper}.

Pilot contamination in multi-user uplinks degrades channel estimates and physical features when users share pilots. Similar problems appear in cell-free massive MIMO \cite{ref87}. Hence, we use a normalized minimum mean-square error (MMSE) combiner, interference-aware features, and causal meta-learning. Related attacks on channel state information (CSI)-based PLA are studied in \cite{zeng2026intelligent}, while pilot contamination in mmWave grant-free massive machine-type communications is addressed in \cite{wang2025joint}.

This paper proposes a secure adaptive framework for authentication in multi-zone networks (SAFA-MZ). The framework combines SCM, backdoor adjustment, invariant risk minimization (IRM), and model-agnostic meta-learning (MAML) to learn environment-invariant authentication features and adapt rapidly to unseen NTN conditions. It also estimates the conditional average treatment effects (CATEs) of key design variables. For efficiency, local multi-feature fingerprint recognition is performed first, and time-difference-of-arrival (TDOA)-based distributed verification is activated only when local authentication is unreliable. In the simulations, SAFA-MZ is compared with representative baselines, including a node-level deep learning method without causal meta-learning or distributed verification, and single-feature authentication schemes that rely on individual physical-layer features and do not employ distributed verification. Finally, physically grounded causal interpretations are provided to explain the influence of environmental factors and design variables on the authentication performance.

The contributions of this paper are summarized as follows:

\begin{itemize}
	
	\item We propose a dual-mode DPLA framework for NTN that uses local multi-feature authentication and activates TDOA-based attack localization only when needed. This reduces overhead and improves robustness and security.
	
	\item We develop a causal-invariant fingerprinting method that uses multiple features under pilot contamination, including directional similarity, beam concentration, combiner consistency, subspace deviation, Doppler-delay characteristics, SINR, and interference-to-noise ratio (INR). This improves robustness against environmental changes and spurious correlations in dynamic NTN channels.
	
	\item We integrate causal representation learning and meta-learning for fast adaptation to new environments with minimal samples. Additionally, we propose a graph-based scheme for coordinated attack detection among aerial nodes and provide causal interpretations for authentication performance.
	
\end{itemize}

\section*{Notation}
We use bold lowercase letters (e.g., \(\mathbf{y}\)) for vectors. 
The Euclidean norm is \(\|\cdot\|\), and the complex modulus is \(|\cdot|\). 
Expectation and probability are denoted by \(\mathbb{E}[\cdot]\) and \(\mathbb{P}(\cdot)\), respectively.
$\mathbf{1}\{\cdot\}$ is the indicator function. 
The symbol \(\circ\) denotes the Hadamard product.
The \(M \times M\) unitary discrete Fourier transform (DFT) matrix is \(\mathbf{F}_M\). 
It transforms the signal from the antenna domain to the angular (beam) domain. 

\section{System Model}
\label{sec:system_model}

As illustrated in Fig.~\ref{fig:system-model}, we consider a set \(\mathcal{K} \triangleq \{k_1, k_2, \dots, k_{|\mathcal{K}|}\}\) of aerial nodes (e.g., high-altitude platform stations (HAPS), LEO, or geostationary orbit (GEO) satellites). Each node \(k\) is equipped with \(M_k\) antennas. The nodes are spatially distributed and connected via a backhaul link to a central fusion center (FC). A set \(\mathcal{U} \triangleq \{u_1, u_2, \dots, u_{|\mathcal{U}|}\}\) of single-antenna users, including legitimate users and illegitimate transmitters (attackers), transmits uplink signals. For a given user \(u\), the subset \(\mathcal{K}_u \subseteq \mathcal{K}\) denotes the nodes that can receive its signal (i.e., the serving nodes). Due to the relative motion between the nodes and the users, the channel experiences severe Doppler shifts and long propagation delays, making DPLA challenging.
Therefore, we consider a two-phase authentication framework. In the local authentication phase, legitimate users are registered at the FC using their physical-layer signatures. During uplink transmission, each observation (either legitimate user or attacker) is locally compared with the enrolled profiles. Matched transmissions are accepted, while suspicious ones are forwarded to the distributed verification phase, where multiple nodes cooperate and the FC makes the final decision.

\begin{figure}[t]
	\centering
	\includegraphics[width=0.98\columnwidth]{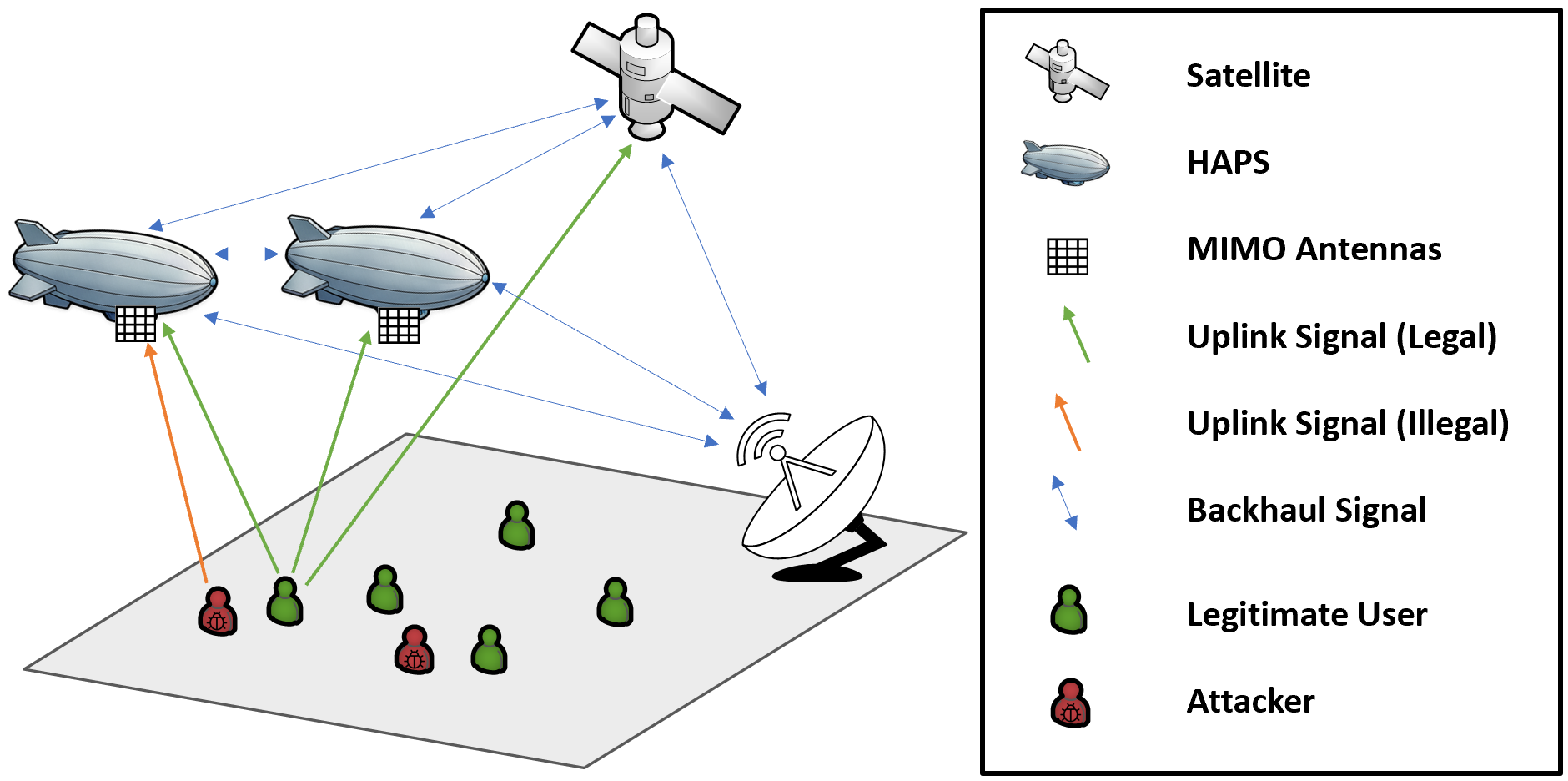}
	\caption{System model with multiple HAPSs/satellites and a fusion center.}
	\label{fig:system-model}
\end{figure}

\subsection{Uplink Signal Model and Feature Extraction}
\label{subsec:signal_model}

The received uplink snapshot at node $k$ over block $n$ is denoted by $\mathbf{y}_{k}(n)\in \mathbb{C}^{M_k}$ and modeled using the following explicit multi-user representation for any target user \(u\):
\begin{align}
	\mathbf{y}_{k}(n)
	=
	&\sum_{u \in \mathcal{U}_k(n)}
	\sqrt{p_{u}(n)}\,\mathbf{h}_{k,u}(n)\,x_{u}(n)
	+
	\mathbf{n}_k(n) \nonumber \\
	=
	&\sqrt{p_{u}(n)}\,\mathbf{h}_{k,u}(n)\,x_{u}(n) \nonumber \\ 
	&+\sum_{\substack{u'\in\mathcal{U}_k(n)\\u'\neq u}}
	\sqrt{p_{u'}(n)}\,\mathbf{h}_{k,u'}(n)\,x_{u'}(n)
	+
	\mathbf{n}_k(n),
	\label{eq:multiuser_rx}
\end{align}
where $\mathcal{U}_k(n)$ denotes the set of users simultaneously active toward node $k$ in block $n$, $p_u(n)$ is the uplink transmit power of user $u$, $\mathbf{h}_{k,u}(n)\in\mathbb{C}^{M_k}$ is the corresponding uplink channel vector, $x_u(n) \in \mathbb{C}$ is the transmitted pilot/data symbol with $\mathbb{E}[|x_u(n)|^2]=1$, and $\mathbf{n}_k(n)\sim\mathcal{CN}(\mathbf{0},\sigma_k^2\mathbf{I}_{M_k})$ is additive noise ($\mathbf{n}_k(n) \in \mathbb{C}^{M_k}$), where $\sigma_k^2$ denotes the noise variance per antenna at node $k$. 
For each active target user $u\in\mathcal{U}_k(n)$, node $k$ applies a user-specific linear combiner $\mathbf{w}_{k,u}(n)\in\mathbb{C}^{M_k}$ to the received vector $\mathbf{y}_{k}(n)$. The resulting scalar output is
\begin{equation}
	y_{k,u}^\mathrm{(out)}(n)=\mathbf{w}_{k,u}^H(n)\mathbf{y}_{k}(n).
\end{equation}
When multiple users are simultaneously active over the same time-frequency resource, the users $u\neq u'$ contribute inter-user interference (IUI); if only one user is active, the interference term vanishes.

Assuming pilot resources are available, let $\hat{\mathbf{h}}_{k,u}(n)\in\mathbb{C}^{M_k}$ denote the estimated uplink channel of active user $u$ at node $k$ in block $n$.
Let $\mathcal{P}_u(n)$ denote the set of users reusing the same pilot as user $u$ in block $n$. Then the pilot-aided channel estimate at node $k$ can be expressed as
\begin{equation}
	\hat{\mathbf{h}}_{k,u}(n)
	=
	\mathbf{h}_{k,u}(n)
	+
	\sum_{\substack{u'\in\mathcal{P}_u(n)\\u'\neq u}}
	\alpha_{u,u'}\,\mathbf{h}_{k,u'}(n)
	+
	\mathbf{e}_{k,u}(n),
	\label{eq:pilot_contamination_model}
\end{equation}
where the second term represents pilot contamination caused by other users sharing the same pilot resource, \(\alpha_{u,u'} \in \mathbb{C}\) denotes the contamination coefficient, and $\mathbf{e}_{k,u}(n)$ collects the residual estimation noise and modeling mismatch.

To suppress both noise and IUI, we adopt a normalized linear MMSE-type combiner
\begin{equation}
	\mathbf{w}_{k,u}(n)
	=
	\frac{\big(\mathbf{R}_y^{(k)}(n)\big)^{-1}\hat{\mathbf{h}}_{k,u}(n)}
	{\left\|\big(\mathbf{R}_y^{(k)}(n)\big)^{-1}\hat{\mathbf{h}}_{k,u}(n)\right\|}.
	\label{eq:mmse_combiner_new}
\end{equation}

Here, the received auto-correlation matrix is defined as
\begin{equation}
	\mathbf{R}_y^{(k)}(n)
	\triangleq
	\mathbb{E}\!\left[\mathbf{y}_{k}(n)\mathbf{y}_{k}^H(n)|\{\mathbf{h}_{k,u}(n)\}_{u\in\mathcal{U}_k(n)}\right],
	\label{eq:Ry_def_new}
\end{equation}
where the expectation is taken over user symbols and noise, conditioned on the instantaneous channel vectors and under mutually uncorrelated user symbols, can be expanded as
\begin{align}
	\mathbf{R}_y^{(k)}(n)=
	&p_{u}(n)\mathbf{h}_{k,u}(n)\mathbf{h}_{k,u}^H(n) \\ \nonumber
	&+
	\sum_{\substack{u'\in\mathcal{U}_k(n)\\u'\neq u}}
	p_{u'}(n)\mathbf{h}_{k,u'}(n)\mathbf{h}_{k,u'}^H(n)
	+
	\sigma_k^2\mathbf{I}_{M_k}.
	\label{eq:Ry_expand}
\end{align}
Therefore, the auto-correlation matrix explicitly captures the desired signal component, the aggregate IUI covariance, and the thermal noise term.
In practice, $\mathbf{R}_y^{(k)}(n)$ is estimated from $N_c$ received samples as
\begin{equation}
	\hat{\mathbf{R}}_y^{(k)}(n)
	=
	\frac{1}{N_c}\sum_{i=1}^{N_c}
	\mathbf{y}_{k}^{(i)}(n)\mathbf{y}_{k}^{(i)H}(n),
	\label{eq:Ry_sample_new}
\end{equation}
where $\mathbf{y}_{k}^{(i)}$ denotes the $i$-th received snapshot at node $k$.
The normalization in \eqref{eq:mmse_combiner_new} ensures $\|\mathbf{w}_{k,u}(n)\|=1$.

Compared with maximum-ratio combining (MRC), the MMSE combiner is more robust in dense multi-user scenarios because it accounts for the spatial covariance of both noise and co-channel interference. Moreover, its computation relies on the instantaneous received statistics and the channel estimate of the target stream, without requiring any higher-layer user identity information.

The uplink channel is modeled as a time-varying Rician fading process. Since the user velocity is assumed negligible relative to the aerial node motion, the dominant Doppler shift is induced by node $k$ and is approximated by
\begin{equation}
	f_{\mathrm d}^{(k,u)}
	=
	\frac{v_k}{c}f_c\cos\!\big(\theta_{\mathrm{elev}}^{(k,u)}\big),
	\label{eq:doppler_new}
\end{equation}
where $v_k$ is the velocity of node $k$, $c=3\times 10^8$~m/s is the speed of light, $f_c$ is the carrier frequency, and $\theta_{\mathrm{elev}}^{(k,u)}\in[0,\pi/2]$ is the elevation angle between node $k$ and user $u$.

Under a narrowband common-Doppler approximation, the uplink channel evolves as
\begin{equation}
	\mathbf{h}_{k,u}(n)
	\approx
	\mathbf{h}_{k,u}^{\mathrm{LoS}}(n)
	e^{j2\pi f_{\mathrm d}^{(k,u)} n T_s}
	+
	\mathbf{h}_{k,u}^{\mathrm{NLoS}}(n),
	\label{eq:rician_channel_new}
\end{equation}
where $\mathbf{h}_{k,u}^{\mathrm{LoS}}(n)\in\mathbb{C}^{M_k}$ denotes the slowly varying line-of-sight component and $\mathbf{h}_{k,u}^{\mathrm{NLoS}}(n)\in\mathbb{C}^{M_k}$ represents the residual diffuse component. 
The dominant deterministic phase rotation is captured through the LoS component, while the residual diffuse component is absorbed into \(\mathbf{h}_{k,u}^{\mathrm{NLoS}}(n)\).
This formulation is especially useful in the present setting because multi-user interference is governed not only by user activity and transmit powers, but also by the temporal evolution and spatial overlap of the corresponding channel vectors.

During the enrollment phase, node $k$ computes a set of reference statistics for each legitimate user $u$ based on $N_{\mathrm{enr}}$ observations. The first quantity is the mean channel vector
\begin{equation}
	\boldsymbol{\mu}_h^{(k,u)}
	\triangleq
	\frac{1}{N_{\mathrm{enr}}}
	\sum_{n=1}^{N_{\mathrm{enr}}}
	\hat{\mathbf{h}}_{k,u}(n),
\end{equation}
where $\boldsymbol{\mu}_h^{(k,u)} \in \mathbb{C}^{M_k}$ represents the average spatial signature of legitimate user $u$ at node $k$, capturing its mean angle-of-arrival structure and average channel gain.
The second quantity is the enrollment channel correlation matrix
\begin{equation}
	\mathbf{R}_h^{(k,u)}
	\triangleq
	\frac{1}{N_{\mathrm{enr}}}
	\sum_{n=1}^{N_{\mathrm{enr}}}
	\hat{\mathbf{h}}_{k,u}(n)\hat{\mathbf{h}}_{k,u}^H(n),
\end{equation}
and the corresponding reference subspace projection matrix is
\begin{equation}
	\mathbf{P}_{\mathrm{ref}}^{(k,u)}
	\triangleq
	\mathbf{U}_{\mathrm{ref}}^{(k,u)}
	\left(\mathbf{U}_{\mathrm{ref}}^{(k,u)}\right)^H,
\end{equation}
where the columns of $\mathbf{U}_{\mathrm{ref}}^{(k,u)} \in \mathbb{C}^{M_k\times r_k}$ are the $r_k$ dominant eigenvectors of $\mathbf{R}_h^{(k,u)}$, with $r_k\leq M_k$. The matrix $\mathbf{P}_{\mathrm{ref}}^{(k,u)}$ captures the dominant spatial subspace occupied by the legitimate user's channel realizations. In NTN scenarios, this subspace may drift slowly due to node motion and geometry variation, and hence periodic re-enrollment may be required.

The reference combiner associated with user $u$ at node $k$ is constructed from the enrollment statistics as
\begin{equation}
	\mathbf{w}_{k,u}^{\mathrm{(ref)}}
	=
	\frac{
		\left(\mathbf{R}_{y,\mathrm{ref}}^{(k)}\right)^{-1}
		\boldsymbol{\mu}_h^{(k,u)}
	}{
		\left\|
		\left(\mathbf{R}_{y,\mathrm{ref}}^{(k)}\right)^{-1}
		\boldsymbol{\mu}_h^{(k,u)}
		\right\|
	},
	\label{eq:ref_combiner}
\end{equation}
where $\mathbf{R}_{y,\mathrm{ref}}^{(k)}$ denotes the received-signal covariance matrix estimated from the enrollment observations of user $u$ at node $k$. This reference combiner is consistent with a normalized linear MMSE-type combiner in \eqref{eq:mmse_combiner_new} and provides a baseline spatial filtering profile for subsequent authentication.

For angular-domain analysis, we employ the unitary DFT matrix $\mathbf{F}_{M_k}\in\mathbb{C}^{M_k\times M_k}$ to transform antenna-domain channel vectors into the beam domain. Let $\mathcal{N}_{\mathrm{main}}^{(k,u)}\subseteq\{1,2,\dots,M_k\}$ denote the set of dominant DFT bins corresponding to the main lobe of legitimate user $u$ as observed at node $k$. These bins are identified during enrollment, for example, as those containing a prescribed fraction of the total beam-domain energy. Since the dominant arrival direction depends on the NTN geometry, $\mathcal{N}_{\mathrm{main}}^{(k,u)}$ may evolve with the elevation angle and should be updated when needed.
Moreover, we define the aggregate post-combining IUI power as
\begin{equation}
	I_{k,u}(n)
	\triangleq
	\sum_{\substack{u'\in\mathcal{U}_k(n)\\u'\neq u}}
	p_{u'}(n)
	\left|
	\mathbf{w}_{k,u}^H(n)\hat{\mathbf{h}}_{k,u'}(n)
	\right|^2 .
	\label{eq:post_comb_interference}
\end{equation}

For each active target user $u\in\mathcal{U}_k(n)$ and each observation block $n$, node $k$ extracts the feature vector
\begin{equation}
	\mathbf{f}_{k,u}(n)
	=
	[f_{k,u,1}(n),\dots,f_{k,u,d_f}(n)]^T,
\end{equation}
where \(d_f=8\) is the feature dimension considered in this paper.
We define the elements of \(\mathbf{f}_{k,u}(n)\in\mathbb{R}^{d_f}\) as follows:

\begin{equation}
	f_{k,u,1}(n)
	=
	\frac{
		\left|
		\hat{\mathbf{h}}_{k,u}^H(n)\boldsymbol{\mu}_h^{(k,u)}
		\right|
	}{
		\|\hat{\mathbf{h}}_{k,u}(n)\|\,
		\|\boldsymbol{\mu}_h^{(k,u)}\|
	}.
	\label{eq:f1_new}
\end{equation}
\(f_{k,u,1}(n)\) is the directional similarity between the current estimated channel and the enrolled mean channel. It measures the alignment of the instantaneous spatial signature with the legitimate reference signature.

\begin{equation}
	f_{k,u,2}(n)
	=
	\frac{
		\sum_{i\in\mathcal{N}_{\mathrm{main}}^{(k,u)}}
		\left|[\mathbf{F}_{M_k}\hat{\mathbf{h}}_{k,u}(n)]_i\right|^2
	}{
		\sum_{i=1}^{M_k}
		\left|[\mathbf{F}_{M_k}\hat{\mathbf{h}}_{k,u}(n)]_i\right|^2
	}.
	\label{eq:f2_new}
\end{equation}
\(f_{k,u,2}(n)\) is the main-lobe concentration ratio in the beam domain. It quantifies how much of the channel energy remains concentrated inside the enrolled dominant angular support.

\begin{equation}
	f_{k,u,3}(n)
	=
	\frac{
		\left|
		\mathbf{w}_{k,u}^H(n)\mathbf{w}_{k,u}^{\mathrm{(ref)}}
		\right|
	}{
		\|\mathbf{w}_{k,u}(n)\|\,
		\|\mathbf{w}_{k,u}^{\mathrm{(ref)}}\|
	}.
	\label{eq:f3_new}
\end{equation}
\(f_{k,u,3}(n)\) is the combiner similarity between the current normalized MMSE combiner and the enrollment-time reference combiner. It reflects the consistency of the spatial filtering direction across time.

\begin{equation}
	f_{k,u,4}(n)
	=
	\frac{
		\left|
		\mathbf{w}_{k,u}^H(n)\mathbf{y}_{k}(n)
		\right|
	}{
		\|\mathbf{w}_{k,u}(n)\|\,
		\|\mathbf{y}_{k}(n)\|
	}.
	\label{eq:f4_new}
\end{equation}
\(f_{k,u,4}(n)\) is a normalized projection score obtained from the combiner output magnitude relative to the received-snapshot norm. It captures the instantaneous compatibility between the combiner and the received signal.

\begin{equation}
	f_{k,u,5}(n)
	=
	1-
	\frac{
		\|\mathbf{P}_{\mathrm{ref}}^{(k,u)}\hat{\mathbf{h}}_{k,u}(n)\|^2
	}{
		\|\hat{\mathbf{h}}_{k,u}(n)\|^2
	}.
	\label{eq:f5_new}
\end{equation}
\(f_{k,u,5}(n)\) is the subspace deviation, namely the fraction of channel energy lying outside the enrolled legitimate-user subspace. Larger values indicate stronger deviation from the enrolled spatial structure.

\begin{equation}
	f_{k,u,6}(n)
	=
	\frac{\exp\!\left(
		-\alpha_{\mathrm d}^{(k,u)}
		\left|
		\hat{\tau}_{k,u}(n)-\tau_{k,u}^{\mathrm{(ref)}}(n)
		\right|
		\right)}{
		1+
		\left|
		\frac{
			\hat{f}_{\mathrm d}^{(k,u)}(n)-f_{\mathrm d,\mathrm{ref}}^{(k,u)}(n)
		}{
			\sigma_\mathrm{f}^{(k,u)}
		}
		\right|^2
	}.
	\label{eq:f6_new}
\end{equation}
\(f_{k,u,6}(n)\) is the Doppler-delay consistency score, which jointly measures the agreement of the observed transmission with the enrolled geometric and kinematic signature of user \(u\). Here, \(\hat{f}_{\mathrm d}^{(k,u)}(n)\) and \(\hat{\tau}_{k,u}(n)\) denote the estimated Doppler shift and propagation delay, respectively, while \(f_{\mathrm d,\mathrm{ref}}^{(k,u)}(n)\) and \(\tau_{k,u}^{\mathrm{(ref)}}(n)\) are their expected reference values.
Moreover, \(\sigma_\mathrm{f}^{(k,u)} > 0\) is a scaling factor for the Doppler mismatch, and \(\alpha_{\mathrm{d}}^{(k,u)} > 0\) is a decay constant that controls how quickly the score drops with increasing delay error.

\begin{equation}
	f_{k,u,7}(n)
	=
	\frac{
		p_u(n)\left|
		\mathbf{w}_{k,u}^H(n)\hat{\mathbf{h}}_{k,u}(n)
		\right|^2
	}{
		I_{k,u}(n)
		+
		\sigma_k^2\|\mathbf{w}_{k,u}(n)\|^2
	}.
	\label{eq:f7_new}
\end{equation}
$f_{k,u,7}(n)$ is the estimated post-combining SINR, measuring the target user's separability from co-channel interferers.

\begin{equation}
	f_{k,u,8}(n)
	=
	\frac{
		I_{k,u}(n)
	}{
		\sigma_k^2\|\mathbf{w}_{k,u}(n)\|^2
	}.
	\label{eq:f8_new}
\end{equation}
$f_{k,u,8}(n)$ denotes the post-combining INR, which measures the aggregate IUI relative to the noise floor and captures interference variations, particularly in dense NTN access scenarios.

Since the normalized MMSE combiner in \eqref{eq:mmse_combiner_new} satisfies $\|\mathbf{w}_{k,u}(n)\|=1$, the terms $\|\mathbf{w}_{k,u}(n)\|^2$ in \eqref{eq:f7_new} and \eqref{eq:f8_new} can be omitted. The first six features measure spatial, subspace, and geometry consistency with the enrolled legitimate profile, while the last two capture interference. Thus, the feature vector is interference-aware and better suited to dense multi-user NTN uplinks than channel-similarity-only features.

\subsection{Authentication}

\subsubsection{Local phase}

At coherence block \(n\), each node \(k\) independently tests observed feature vector \(\mathbf{f}_{k,u}(n)\) against enrolled user identity \(u'\) for local authentication.
The local soft score, interpreted as a legitimacy score with respect to user \(u'\), is computed from the Mahalanobis distance as
\begin{align}
	&d_{M}\big(\mathbf{f}_{k,u}(n);u'\big)
	= \nonumber \\
	&\quad\sqrt{
		\big(\mathbf{f}_{k,u}(n)-\boldsymbol{\mu}_{f}^{(k,u')}\big)^T
		\big(\boldsymbol{\Sigma}_{f}^{(k,u')}\big)^{-1}
		\big(\mathbf{f}_{k,u}(n)-\boldsymbol{\mu}_{f}^{(k,u')}\big)
	},
	\label{eq:local_mahal}
\end{align}
where \(\boldsymbol{\mu}_{f}^{(k,u')} \in \mathbb{R}^{d_f}\) and
\(\boldsymbol{\Sigma}_{f}^{(k,u')} \in \mathbb{R}^{d_f \times d_f}\) denote the mean vector and covariance matrix of the legitimate feature distribution for enrolled user \(u'\) at node \(k\), respectively.
During enrollment, \(N_{\mathrm{enr}}\) legitimate feature snapshots are collected for each enrolled user \(u'\) at node \(k\).
Based on these enrollment samples, the sample mean and covariance are estimated as
\begin{align}
	&\hat{\boldsymbol{\mu}}_{f}^{(k,u')}
	=
	\frac{1}{N_{\mathrm{enr}}}
	\sum_{i=1}^{N_{\mathrm{enr}}}
	\mathbf{f}_{k,u'}(i),
	\label{eq:mean_est}
	\\
	&\hat{\boldsymbol{\Sigma}}_{f}^{(k,u')}
	= \nonumber \\
	&\quad\frac{1}{N_{\mathrm{enr}}-1}
	\sum_{i=1}^{N_{\mathrm{enr}}}
	\big(\mathbf{f}_{k,u'}(i)-\hat{\boldsymbol{\mu}}_{f}^{(k,u')}\big)
	\big(\mathbf{f}_{k,u'}(i)-\hat{\boldsymbol{\mu}}_{f}^{(k,u')}\big)^T.
	\label{eq:cov_est}
\end{align}
In practice, \(d_M(\mathbf{f}_{k,u}(n);u')\) is evaluated by replacing \(\boldsymbol{\mu}_{f}^{(k,u')}\) and \(\boldsymbol{\Sigma}_{f}^{(k,u')}\) with their sample estimates in \eqref{eq:mean_est}--\eqref{eq:cov_est}.

The Mahalanobis distance normalizes the feature space using the spread and correlation structure of legitimate samples. Hence, it is scale-invariant and accounts for inter-feature correlations. A small value of $d_M\big(\mathbf{f}_{k,u}(n);u'\big)$ indicates strong consistency with the legitimate profile of user $u'$, whereas a large value suggests a potential attacker or identity mismatch.
This distance is then mapped to a local soft score defined as
\begin{equation}
	p_{k,u,u'}(n)
	\triangleq
	\frac{1}{1 + e^{\,d_M(\mathbf{f}_{k,u}(n);u') - \tau_k}},
	\label{eq:local_soft}
\end{equation}
where \(\tau_k > 0\) is the local decision threshold at node \(k\).
Moreover, \(p_{k,u,u'}(n)\approx 1\) when \(d_M\big(\mathbf{f}_{k,u}(n);u'\big)\ll \tau_k\), corresponding to a close match with the legitimate profile of user \(u'\), whereas \(p_{k,u,u'}(n)\approx 0\) when \(d_M\big(\mathbf{f}_{k,u}(n);u'\big)\gg \tau_k\), indicating an attacker or an identity mismatch.

\subsubsection{Distributed Verification and Attack Detection}
\label{sec:Adapt_Distibuted_Verification}
The proposed framework operates in two modes.
In the local authentication mode, each node independently compares the observed feature vector against the enrolled legitimate-user profiles. We define $\mathcal{U}_{\mathrm{enr}}$ as the set of enrolled users.
For an observation \(\mathbf{f}_{k,u}(n)\) at node \(k\), let
\begin{equation}
	u_k^\star(n)
	=
	\arg\min_{u' \in \mathcal{U}_{\mathrm{enr}}}
	d_M\big(\mathbf{f}_{k,u}(n);u'\big),
	\label{eq:best_identity}
\end{equation}
denote the best-matching enrolled identity, and define the corresponding minimum distance as
\begin{equation}
	d_{k,\min}(n)
	=
	\min_{u' \in \mathcal{U}_{\mathrm{enr}}}
	d_M\big(\mathbf{f}_{k,u}(n);u'\big).
	\label{eq:min_dist}
\end{equation}
Let \(\tau_{\mathrm{ac}} > 0\) denote the local acceptance threshold.
If \(d_{k,\min}(n) \le \tau_{\mathrm{ac}}\), the observation at node \(k\) is considered locally consistent with the enrolled profile of user \(u_k^\star(n)\).
However, if there exists at least one participating node \(k \in \mathcal{K}_u\) such that $d_{k,\min}(n) > \tau_{\mathrm{ac}}$, the transmission is flagged as suspicious and the system switches to the adaptive distributed verification mode.

In the distributed verification (attack detection) mode, multiple nodes collaboratively contribute to the final authentication decision.
The FC collects from each participating node \(k \in \mathcal{K}_u\) the local soft score
\(p_{k,u,u_k^\star}(n)\), the associated best-matching identity \(u_k^\star(n)\), and the time-of-arrival (TOA) estimate \(\hat{\tau}_{k,u}(n)\).
The TOA measurement is modeled as
\begin{equation}
	\hat{\tau}_{k,u}(n)
	=
	\frac{\|\mathbf{r}_u(n) - \mathbf{r}_k(n)\|}{c}
	+
	\epsilon_{k,u}(n),
	\label{eq:toa}
\end{equation}
where \(\mathbf{r}_u(n) \in \mathbb{R}^3\) is the unknown transmitter position, \(\mathbf{r}_k(n) \in \mathbb{R}^3\) is the known position of node \(k\), \(c\) is the speed of light, and \(\epsilon_{k,u}(n)\) denotes the TOA estimation error. 
A model for $\epsilon_{k,u}(n)$ is zero-mean Gaussian noise with variance $\sigma_{\mathrm{toa}}^2$.
We assume that the aerial nodes are time-synchronized through the backhaul network.
Using the TOA measurements collected from the nodes in \(\mathcal{K}_u\), the FC forms TDOA observations and computes a position estimate \(\hat{\mathbf{r}}_u(n)\) for the transmitter.
Based on this estimate, two location-consistency scores are constructed.
The first score quantifies the stability of the estimated position over time:
\begin{align}
	f_{\mathrm{tri},u}(n)
	&=
	\exp\!\left(
	-\frac{
		\operatorname{Tr}\,\bigl(
		\widehat{\mathrm{Cov}}(\hat{\mathbf{r}}_u(n))
		\bigr)
	}{
		2\sigma_{\mathrm{pos}}^2
	}
	\right),
	\label{eq:tri_score}
\end{align}
where \(\widehat{\mathrm{Cov}}(\hat{\mathbf{r}}_u(n))\) denotes the empirical covariance of recent position estimates and \(\sigma_{\mathrm{pos}}^2\) is a scaling parameter.
A small value of \(\sigma_{\mathrm{pos}}\) makes the score more sensitive to temporal fluctuations in the estimated position.

The second score measures the consistency between the estimated transmitter location and the reference location associated with the claimed or best-matching identity:
\begin{align}
	f_{\mathrm{dist},u,u^\star}(n)
	&=
	\exp\!\left(
	-\frac{
		\|\hat{\mathbf{r}}_u(n)-\mathbf{r}_{u^\star}^{\mathrm{(ref)}}(n)\|^2
	}{
		2\sigma_{\mathrm{dist}}^2
	}
	\right),
	\label{eq:dist_score}
\end{align}
where \(\mathbf{r}_{u^\star}^{\mathrm{(ref)}}(n)\) is the reference or claimed location associated with identity \(u^\star\), and \(\sigma_{\mathrm{dist}}^2\) controls the tolerance to spatial mismatch.
A smaller \(\sigma_{\mathrm{dist}}\) imposes a stricter consistency requirement between the estimated and claimed positions.

At the FC, the final decision is obtained through a scalable graph-based fusion mechanism rather than direct feature concatenation.
Specifically, among all participating nodes in \(\mathcal{K}_u\), the FC first identifies the \(L\) most reliable nodes based on their local verification evidence.
It then constructs a graph over the selected nodes using delay, Doppler, and geometric consistency information, performs graph message passing, and generates a graph-level representation through attention-based readout.
Finally, the global triangulation-consistency and identity-location consistency scores are appended to the graph-level representation before the final output layer.
The detailed formulation is provided in Section~\ref{subsec:topL_graph_fusion}.

\subsection{Environment and Design Variables for NTN}

The environmental and link-level context for user \(u\) observed at node \(k\) is represented by
\begin{equation}
	Z_{k,u}
	\triangleq
	\bigl\{
	\gamma_{k,u},\;
	M_k,\;
	\sigma_{\theta,k},\;
	K_\mathrm{r}^{(k,u)},\;
	H_\mathrm{alt}^{(k)},\;
	\theta_\mathrm{elev}^{(k,u)},\;
	v_k
	\bigr\}.
	\label{eq:envZ}
\end{equation}
Here, \(\gamma_{k,u}\) denotes the received SNR of user \(u\) at node \(k\), expressed in dB;
\(\sigma_{\theta,k}\) is the angular spread observed at node \(k\), which is typically small in NTN links dominated by strong LoS components;
\(K_\mathrm{r}^{(k,u)}\) is the Rician factor of the link between user \(u\) and node \(k\);
\(H_\mathrm{alt}^{(k)}\in\mathbb{R}_{+}\) is the altitude of node \(k\), measured in km.
The variables in \(Z_{k,u}\) characterize the operating environment and the link geometry.
In particular, low elevation angles generally lead to longer propagation distances and may increase path loss and Doppler sensitivity, while higher node velocities increase the Doppler shift according to \eqref{eq:doppler_new}.
Moreover, the Rician factor, angular spread, pilot reuse factor, and analog-to-digital converter (ADC) resolution affect the quality of the extracted physical-layer features and, consequently, the authentication performance.

We distinguish these environmental variables from controllable design variables.
Let $\mathbf{T} \in \mathcal{T}$ denote the treatment, i.e., the controllable system configuration selected by the designer, where \(\mathcal{T}\) is the set of admissible treatment values.
Depending on the experiment, \(\mathbf{T}\) may represent a scalar design choice, such as the ADC resolution, \(\mathbf{T}=N_\mathrm{bit}^{(k)}\),
the pilot reuse factor, \(\mathbf{T}=K_\mathrm{p}^{(k)}\),
the top-\(L\) selection budget, \(\mathbf{T}=L\),
or the node orbit/platform type, e.g., HAPS, LEO, or GEO.
More generally, \(\mathbf{T}\) may be a vector of design parameters
\begin{equation}
	\mathbf{T}
	=
	\bigl[
	N_\mathrm{bit}^{(k)},\;
	K_\mathrm{p}^{(k)},\;
	L,\;
	\mathrm{orbit}^{(k)}
	\bigr]^T.
	\label{eq:treatment_vector}
\end{equation}
Each realization of \(\mathbf{T}\) corresponds to a particular system configuration whose effect on the authentication performance metrics, such as TPR, FPR, AUC, or detection accuracy, can be evaluated.
Unlike the environmental variables \(Z_{k,u}\), which describe the propagation, mobility, and hardware context experienced by the link, the treatment variables represent design choices that can be actively controlled by the system designer.
Depending on the considered NTN deployment, such treatments can be configured either globally at the network level or locally on a per-node basis.

\section{Proposed Method}

Our method has three main goals in NTN. First, it estimates how design choices in \(\mathbf{T}\) causally affect authentication under severe Doppler, long delays, and fast channel variations, using causal learning to separate true effects from spurious correlations and account for environmental confounders. This also reveals whether the effect is direct or mediated through physical-layer features. Second, it trains a meta-learner that quickly adapts to new NTN environments, while incorporating the learned causal knowledge to improve robustness and consistency. Third, if the meta-learner cannot authenticate a user, a TDOA-based attack detection module is activated. A graph attention network then uses only the top-\(L\) most reliable nodes. This improves scalability and reduces complexity while preserving informative spatial observations.

\subsection{SCM for NTN}

The feature vector is obtained via~\eqref{eq:f1_new}--\eqref{eq:f8_new}, yet we need to discover the causal relationships between the environment, treatments, and the feature vector. For each node \(k\) and user \(u\), the data-generating process is given by the SCM:
\begin{equation}
	\mathbf{f}_{k,u}(n) = g_{f,k}(\mathbf{T}, Z_{k,u}, U_{f,k}), 
	\label{eq:scm_ntn}
\end{equation}
where \(g_{f,k}(\cdot)\) is an unknown deterministic function representing the underlying structural mechanisms, which are approximated in practice using neural network models in \eqref{eq:phi}. 
The variable \(U_{f,k}\) denotes exogenous random variable that is assumed to be mutually independent of \(\mathbf{T}\) and \(Z_{k,u}\).
It captures all unobserved sources of randomness affecting the feature extraction and authentication outcome, including thermal noise, channel estimation errors, and hardware impairments. 
The unknown function \(g_{f,k}\) in \eqref{eq:scm_ntn} is implemented indirectly through the neural network components \(\Phi_{\phi}\) for node \(k\).
The structural mechanisms are learned indirectly through the neural components.
Hence, we use a neural network as a feature extractor \(\Phi_{\phi}\) which maps the raw input \(\mathcal{X}_{k,u}(n)\) to an approximated feature vector \(\widehat{\mathbf{f}}_{k,u}(n)\). 
\begin{equation}
	\widehat{\mathbf{f}}_{k,u}(n) = \Phi_{\phi}(\mathcal{X}_{k,u}(n)),
	\label{eq:phi}
\end{equation}
The raw input \(\mathcal{X}_{k,u}(n)\) includes \(\mathbf{y}_{k}(n)\), \(\hat{\mathbf{h}}_{k,u}(n)\), \(\hat{f}_{\mathrm d}^{(k,u)}(n)\), and \(\hat{\tau}_{k,u}(n)\).
In summary, in reality the features are obtained via the relation in \eqref{eq:f1_new}--\eqref{eq:f8_new}. These can be considered as an unknown function of \(\mathbf{T}\), \(Z_{k,u}\), and \(U_{f,k}\), as defined in \eqref{eq:scm_ntn}. Thus, we estimate this function using \eqref{eq:phi}. 
Generally, if the approximated feature \(\hat{\mathbf{f}}_{k,u}(n)\) is similar to extracted feature \(\mathbf{f}_{k,u}(n)\), the transmission is more likely to be authenticated as legitimate.
Furthermore, we find the identity by~\eqref{eq:best_identity}.

We use \(Y_{k,u}(n)\in\{0,1\}\) as the binary local authentication outcome at node \(k\), where \(Y_{k,u}(n)=0\) indicates that the observed transmission is authenticated as legitimate and \(Y_{k,u}(n)=1\) otherwise.

In practice, we produce a local authentication outcome score \(\widehat{Y}_{k,u}(n)\) at node \(k\) with the method in~\ref{sec:Adapt_Distibuted_Verification} by using estimated feature vector \(\hat{\mathbf{f}}_{k,u}(n)\) instead of \(\mathbf{f}_{k,u}(n)\):
\begin{equation}
	\widehat{Y}_{k,u}(n) = 
	\begin{cases}
		0, & \text{if } \; d_{k,\min}(n) \le \tau_{\mathrm{ac}} , \\
		1, & \text{otherwise}.
	\end{cases}
\end{equation}

As illustrated in Fig.~\ref{fig:SCM}, the corresponding causal graph includes the directed edges:
\(Z_{k,u} \to \mathbf{f}_{k,u}\), \(Z_{k,u} \to Y_{k,u}(n)\), \(\mathbf{T} \to \mathbf{f}_{k,u}(n)\), \(\mathbf{T} \to Y_{k,u}(n)\), and \(\mathbf{f}_{k,u}(n) \to Y_{k,u}(n)\). 
In the assumed causal graph, 
both \(Z_{k,u}\) and \(\mathbf{T}\) act as confounders for the causal relationship 
\(\mathbf{f}_{k,u}(n) \to Y_{k,u}(n)\), since each directly affects both the cause (\(\mathbf{f}_{k,u}(n)\)) and the outcome (\(Y_{k,u}(n)\)).
We further allow statistical dependence between \(Z_{k,u}\) and \(\mathbf{T}\), reflecting that system design choices may depend on the operating environment. Given the causal graph (Fig.~\ref{fig:SCM}), two causal relationships must be distinguished:

\begin{enumerate}
	\item For the effect of \(\mathbf{f}_{k,u}(n)\) on \(Y_{k,u}(n)\):
	Both \(Z_{k,u}\) and \(\mathbf{T}\) are direct confounders. 
	Hence, to identify the causal effect \(\mathbf{f}_{k,u}(n) \to Y_{k,u}(n)\), 
	we must condition on both \(Z_{k,u}\) and \(\mathbf{T}\) (i.e., include them as control variables). 
	Failure to condition on either one opens a backdoor path and biases the estimate.
	
	\item For the effect of \(\mathbf{T}\) on \(Y_{k,u}(n)\): 
	There is no direct edge from \(Z_{k,u}\) to \(\mathbf{T}\); the statistical dependence between them arises from an unobserved confounder \(U_{f,k}\). 
	This creates the backdoor path \(\mathbf{T} \leftarrow U_{f,k} \to Z_{k,u} \to Y_{k,u}(n)\). 
	Conditioning on \(Z_{k,u}\) alone blocks this path because \(Z_{k,u}\) is a descendant of \(U_{f,k}\) and a parent of \(Y\). 
	Therefore, for the causal effect \(\mathbf{T} \to Y_{k,u}(n)\), it is sufficient to condition only on \(Z_{k,u}\), not on \(\mathbf{T}\) itself (as that would condition on the treatment and introduce bias). 
	The unobserved confounder \(U_{f,k}\) does not need to be measured; its backdoor path is blocked by the observable \(Z_{k,u}\).
\end{enumerate}

\begin{figure}[t]
	\centering
	\includegraphics[width=0.40\columnwidth]{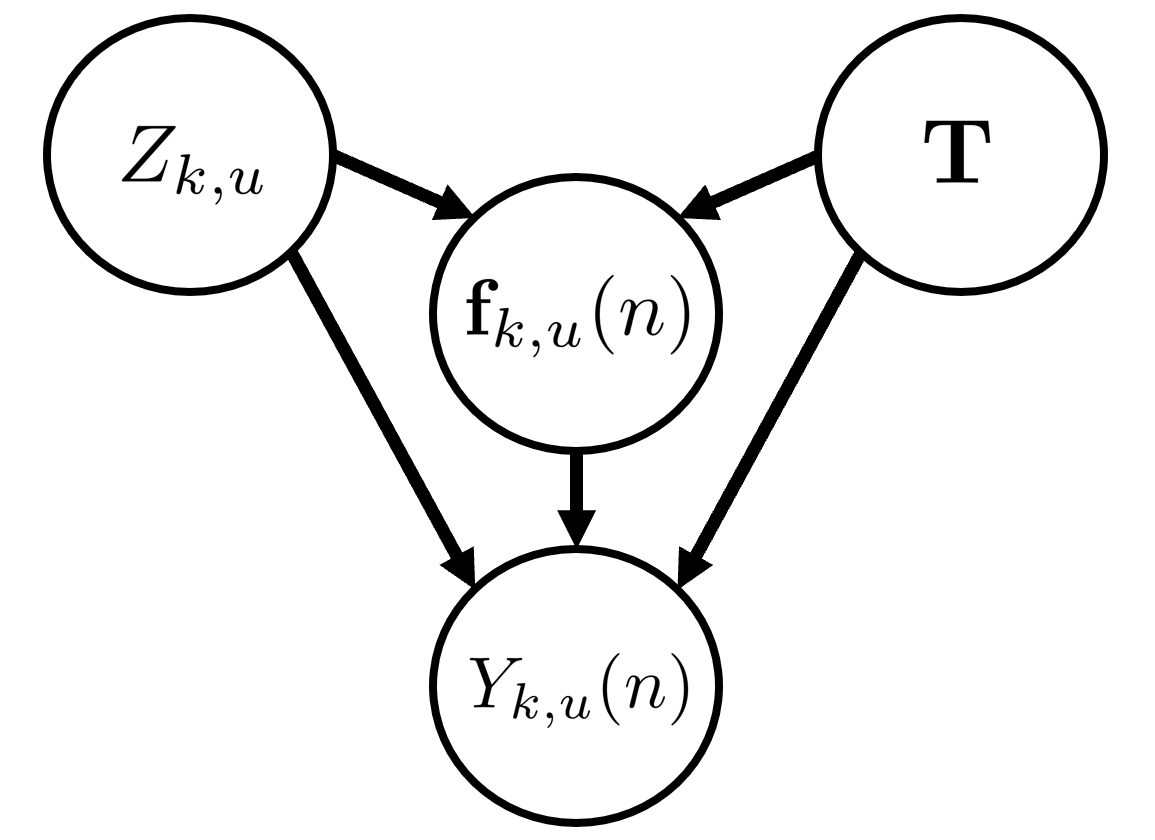}
	\caption{Causal graph of the data-generating process.}
	\label{fig:SCM}
\end{figure}

In NTN scenarios, there exist unobserved confounders (e.g., network congestion, operator policies) that affect both the treatment \(\mathbf{T}\) and the environmental variables \(Z_{k,u}\). These unobserved confounders create spurious backdoor paths. Although \(Z_{k,u}\) themselves are not direct confounders (they have no edge to \(\mathbf{T}\)), they act as sufficient adjustment variables because conditioning on them blocks all backdoor paths between \(\mathbf{T}\) and \(Y_{k,u}(n)\). Therefore, we assume that the set \(Z_{k,u}\) satisfies the backdoor criterion relative to \((\mathbf{T}, Y_{k,u}(n))\). Furthermore, we assume positivity (overlap) and the stable unit treatment value assumption (SUTVA). Under SUTVA and the backdoor criterion, the causal effect of \(\mathbf{T}\) on \(Y_{k,u}(n)\) is identifiable from observational NTN data.

Let \(t_0, t_1 \in \mathcal{T}\) be two treatment values of interest.
For each node \(k\) and user \(u\), the interventional mean is defined as
\begin{equation}
	\mu_{k,u}(t) \triangleq \mathbb{E}\big[\widehat{Y}_{k,u}(n) \mid \mathrm{do}(\mathbf{T}=t)\big],
	\label{eq:muY}
\end{equation}
and the corresponding average causal effect (ACE) is
\begin{equation}
	\mathrm{ACE}_{k,u}(t_1, t_0) 
	\triangleq \mu_{k,u}(t_1) - \mu_{k,u}(t_0).
	\label{eq:ace}
\end{equation}

Under the backdoor assumption with respect to \(Z_{k,u}\), the interventional mean can be identified via
\begin{equation}
	\mu_{k,u}(t) 
	= \mathbb{E}_{Z_{k,u}}\Big[\mathbb{E}\big[\widehat{Y}_{k,u}(n) \mid \mathbf{T}=t, Z_{k,u}\big]\Big].
	\label{eq:backdoor}
\end{equation}

Under the standard consistency, conditional exchangeability, and positivity assumptions, the interventional mean \(\mu_{k,u}(t)\) is identifiable from observational data via the backdoor formula; the proof is provided in Appendix A-A of the Supplemental Material.

The conditional average treatment effect (CATE) at a given environment realization \(z_{k,u}\) is defined as
\begin{align}
	&\mathrm{CATE}_{k,u}(t_1, t_0 \mid z_{k,u}) \triangleq \nonumber \\ 
	&\quad \mathbb{E}\big[\widehat{Y}_{k,u}(n) \mid \mathrm{do}(\mathbf{T}=t_1), Z_{k,u}=z_{k,u}\big] \nonumber \\
	&\quad-\mathbb{E}\big[\widehat{Y}_{k,u}(n) \mid \mathrm{do}(\mathbf{T}=t_0), Z_{k,u}=z_{k,u}\big].
	\label{eq:cate}
\end{align}
Under the backdoor condition, this can be written as
\begin{align}
	&\mathrm{CATE}_{k,u}(t_1, t_0 \mid z_{k,u}) = \nonumber \\
	&\quad \mathbb{E}\big[\widehat{Y}_{k,u}(n) \mid \mathbf{T}=t_1, Z_{k,u}=z_{k,u}\big] \nonumber \\ 
	&\quad - \mathbb{E}\big[\widehat{Y}_{k,u}(n) \mid \mathbf{T}=t_0, Z_{k,u}=z_{k,u}\big].
	\label{eq:cate_2}
\end{align}

In NTN scenarios, CATE analysis is particularly useful for characterizing how the effectiveness of a given system design varies with environmental factors such as the elevation angle \(\theta_{\mathrm{elev}}^{(k,u)}\) and node velocity \(v_k\).

Let $\mathbf{f}_{k,u}(t')$ denote the potential feature vector at node $k$ for user $u$ under treatment $T=t'$ at time index \(n\). For notational simplicity, the time index $n$ is omitted throughout this section. This vector serves as the mediator according to the causal graph. Let $\widehat{Y}_{k,u}(t,\mathbf{f}_{k,u}(t'))$ denote the model-predicted potential outcome at time index $n$ when the treatment is set to $t$ and the mediator is set to $\mathbf{f}_{k,u}(t')$. We also omit \(n\) from $\widehat{Y}_{k,u}(t,\mathbf{f}_{k,u}(t'))$ for simplicity.

Thus, the natural direct effect (NDE) is defined as
\begin{align}
	\mathrm{NDE}_{k,u}(t_1, t_0)
	=&\ \mathbb{E}\!\left[
	\widehat{Y}_{k,u}\bigl(t_1,\mathbf{f}_{k,u}(t_0)\bigr)
	\right] \nonumber \\
	&-\mathbb{E}\!\left[
	\widehat{Y}_{k,u}\bigl(t_0,\mathbf{f}_{k,u}(t_0)\bigr)
	\right],
	\label{eq:nde}
\end{align}
and the natural indirect effect (NIE) is given by
\begin{align}
	\mathrm{NIE}_{k,u}(t_1, t_0)
	=&\ \mathbb{E}\!\left[
	\widehat{Y}_{k,u}\bigl(t_0,\mathbf{f}_{k,u}(t_1)\bigr)
	\right] \nonumber \\
	&-\mathbb{E}\!\left[
	\widehat{Y}_{k,u}\bigl(t_0,\mathbf{f}_{k,u}(t_0)\bigr)
	\right],
	\label{eq:nie}
\end{align}
The total effect is decomposed as
\begin{equation}
	\mathrm{ACE}_{k,u}(t_1, t_0)
	= \mathrm{NDE}_{k,u}(t_1, t_0)
	+ \mathrm{NIE}_{k,u}(t_1, t_0),
	\label{eq:ace_decomp}
\end{equation}
which is formally established in Appendix A-C of the Supplemental Material.

Given \(N_s\) samples collected from multiple NTN environments and nodes, we denote the dataset as
\(\{(Z_{k,u}^{(i)}, \mathbf{T}^{(i)}, \mathbf{f}_{k,u}^{(i)}, \widehat{Y}_{k,u}^{(i)}(n))\}_{i=1}^{N_s}\),
where each sample corresponds to an observation at node \(k\), user \(u\), sample \(i\), and time index \(n\). To simplify the formulation we omit time index \(n\).
We define the nuisance models as
\begin{equation}
	e_t(z) \triangleq \mathbb{P}(\mathbf{T}^{(i)}=t \mid Z_{k,u}^{(i)}=z),
	\label{eq:propensity}
\end{equation}
\begin{equation}
	m_t(z) \triangleq \mathbb{E}[\widehat{Y}_{k,u}^{(i)}(n) \mid \mathbf{T}^{(i)}=t, Z_{k,u}^{(i)}=z].
	\label{eq:outcome}
\end{equation}

The doubly robust (DR) estimator (also known as augmented inverse probability weighting, AIPW) of the interventional mean
\(\mu_{k,u}(t)\)
is given by
\begin{align}
	&\widehat{\mu}_{k,u}^{\mathrm{DR}}(t) 
	= \nonumber \\
	& \quad \frac{1}{N}\sum_{i=1}^{N}\left[
	m_t(Z_{k,u}^{(i)})
	+ \frac{\mathbf{1}\{\mathbf{T}^{(i)}=t\}}{e_t(Z_{k,u}^{(i)})}
	\big(\widehat{Y}_{k,u}^{(i)}(n) - m_t(Z_{k,u}^{(i)})\big)
	\right].
	\label{eq:dr}
\end{align}

This estimator is consistent under standard ignorability, overlap, and correct specification of either the propensity model \(e_t(z)\) or the outcome model \(m_t(z)\); the proof is given in Appendix A-B of the Supplemental Material.
For comparison, the inverse probability of treatment weighting (IPTW) estimator is defined as
\begin{equation}
	\widehat{\mu}_{k,u}^{\mathrm{IPTW}}(t) 
	= \frac{1}{N_s}\sum_{i=1}^{N_s} 
	\frac{\mathbf{1}\{\mathbf{T}^{(i)}=t\}\,\widehat{Y}_{k,u}^{(i)}(n)}{e_t(Z_{k,u}^{(i)})}.
	\label{eq:iptw}
\end{equation}

\subsection{Meta-Learning with Causal Invariance for NTN}

Let $\mathcal{E} = \{e_1, e_2, \dots, e_{|\mathcal{E}|}\}$ denote the set of meta-training NTN environments. 
Each environment $e \in \mathcal{E}$ corresponds to a distinct operating condition characterized by context variables $Z_{k,u}^{e}$. For each environment $e$ and observing node $k$, we define a meta-learning task indexed by the pair $(e,k)$, with training dataset $D_{\mathrm{train}}^{e,k}$ and validation dataset $D_{\mathrm{val}}^{e,k}$. Each dataset contains samples from multiple users.
We use a shared authentication model parameterized by $\phi$, as defined in \eqref{eq:scm_ntn}. $\Phi_{\phi}$ is trained to capture causal and environment-invariant structure across tasks.
For each task $(e,k)$, one inner-loop adaptation step is performed as
\begin{equation}
	\phi'_{e,k}
	=
	\phi
	-
	\alpha \nabla_{\phi}
	\mathcal{L}_{\mathrm{task}}^{e,k}(\phi;D_{\mathrm{train}}^{e,k}),
	\label{eq:maml_inner}
\end{equation}
where $\alpha$ is the inner-loop learning rate and $\mathcal{L}_{\mathrm{task}}^{e,k}$ denotes the task-specific authentication loss.
The meta-parameters are then updated using the validation losses across tasks:
\begin{equation}
	\phi \leftarrow \phi - \beta \sum_{e,k} \nabla_{\phi}
	\mathcal{L}_{\mathrm{task}}^{e,k}(\phi'_{e,k};D_{\mathrm{val}}^{e,k}),
	\label{eq:maml_outer}
\end{equation}
where $\beta$ is the outer-loop learning rate.

To encourage $\Phi_{\phi}$ to extract environment-invariant causal features, we add an invariant risk minimization (IRM) penalty:
\begin{equation}
	\mathcal{R}_{\mathrm{IRM}}(\phi)
	=
	\sum_{e,k}
	\left\|
	\nabla_{w}\, R^{e,k}(w,\phi)\big|_{w=1}
	\right\|^2,
	\label{eq:irm}
\end{equation}
where $R^{e,k}(w,\phi)$ denotes the empirical risk for task $(e,k)$ under a scalar classifier parameter $w$ used only to construct the IRM penalty.
If the learned representation is invariant across environments, the IRM penalty vanishes; this statement is formally proved in Appendix A-D of the Supplemental Material.

To further promote causal stability, we regularize the task-specific DR estimates of the interventional mean. Let $\widehat{\mu}_{e,k}^{\mathrm{DR}}(t)$ denote the DR estimate of the interventional mean under treatment $t$ computed from task $(e,k)$, and let
\begin{equation}
	\bar{\mu}_{-e,-k}^{\mathrm{DR}}(t)
	=
	\frac{1}{|\mathcal{I}|-1}
	\sum_{(e',k')\in\mathcal{I} \setminus \{(e,k)\}}
	\widehat{\mu}_{e',k'}^{\mathrm{DR}}(t),
\end{equation}
where $\mathcal{I}=\{(e,k):e\in\mathcal{E}, k\in\mathcal{K}\}$ is the set of all task indices. We define the DR consistency regularizer as
\begin{equation}
	\mathcal{R}_{\mathrm{DR}}
	=
	\sum_{t\in\mathcal{T}}
	\sum_{(e,k)\in\mathcal{I}}
	\left(
	\widehat{\mu}_{e,k}^{\mathrm{DR}}(t)
	-
	\bar{\mu}_{-e,-k}^{\mathrm{DR}}(t)
	\right)^2.
	\label{eq:dr_reg}
\end{equation}
This term encourages, but does not strictly enforce, stability of interventional responses across environments and nodes.

The final meta-training objective is
\begin{equation}
	\min_{\phi}
	\sum_{e,k}
	\mathcal{L}_{\mathrm{task}}^{e,k}\!\left(
	\phi'_{e,k};D_{\mathrm{val}}^{e,k}
	\right)
	+ \lambda_{\mathrm{IRM}}\,\mathcal{R}_{\mathrm{IRM}}(\phi)
	+ \lambda_{\mathrm{DR}}\,\mathcal{R}_{\mathrm{DR}},
	\label{eq:meta_obj}
\end{equation}
where $\lambda_{\mathrm{IRM}}$ and $\lambda_{\mathrm{DR}}$ control the strength of the invariance and causal-consistency regularization terms, respectively.

\subsection{Counterfactual Design Policy for NTN}

We learn a policy \(\pi: Z_{k,u} \mapsto \mathcal{T}\) that selects the design configuration for each environment as observed by node \(k\). Thus, for a given \(Z_{k,u}\), the applied treatment is \(\pi(Z_{k,u}) \in \mathcal{T}\).
The goal is to maximize the interventional performance under the induced policy, defined as \(J(\pi) = \mathbb{E}_{Z_{k,u}}\big[ \mu_{k,u}(\pi(Z_{k,u})) \big]\),
where \(\mu_{k,u}(t)\) is the interventional mean defined in \eqref{eq:muY}.

We estimate \(J(\pi)\) from logged data collected from all nodes using the DR off-policy estimator:
\begin{align}
	&\widehat{J}_{\mathrm{DR}}(\pi) = \nonumber \\
	&\quad \frac{1}{N_s}\sum_{i=1}^{N_s} m_{\pi(Z_{k,u}^{(i)})}(Z_{k,u}^{(i)})  \nonumber \\
	&\quad + \frac{1}{N_s}\sum_{i=1}^{N_s} \frac{\mathbf{1}\{\mathbf{T}^{(i)} = \pi(Z_{k,u}^{(i)})\}}{e_{\mathbf{T}^{(i)}}(Z_{k,u}^{(i)})}
	\big( \widehat{Y}_{k,u}^{(i)}(n) - m_{\mathbf{T}^{(i)}}(Z_{k,u}^{(i)}) \big),
	\label{eq:ope}
\end{align}
where \(m_t(\cdot)\) is the outcome model in \eqref{eq:outcome}, \(e_t(\cdot)\) is the propensity model in \eqref{eq:propensity}.
This estimator is DR, i.e., it remains consistent if either the outcome model or the propensity model is correctly specified.
Its off-policy nature is especially valuable in NTN, as it enables policy evaluation and optimization without live deployment, thereby reducing the cost and operational risk of HAPS and satellite platforms.

\subsection{Scalable Top-\(L\) Graph-Based Fusion at the FC}
\label{subsec:topL_graph_fusion}

In attack detection mode, the FC collects local authentication evidence from nodes in $\mathcal{K}_u$ and applies a top-$L$ graph-based fusion strategy. For each node $k\in\mathcal{K}_u$, it first forms a node-level evidence vector as
\begin{equation}
	\mathbf{x}_{k,u}(n)
	\triangleq
	\Bigl[
	p_{k,u,u_k^\star}(n),\,
	\mathbf{f}_{k,u}^T(n),\,
	\hat{\tau}_{k,u}(n),\,
	\hat{f}_{\mathrm d}^{(k,u)}(n),\,
	\mathbf{r}_k^T(n)
	\Bigr]^T,
	\label{eq:fc_node_evidence}
\end{equation}
where \(\mathbf{x}_{k,u}(n) \in \mathbb{R}^{d_x}\) and \(d_x \triangleq d_f+6\).
The FC evaluates the reliability of the local evidence from node \(k\) through a trainable function \(\rho_{k,u}(n)
\triangleq
g_{\theta_{\rho}}\!\bigl(\mathbf{x}_{k,u}(n)\bigr)\),
where $g_{\theta_{\rho}}$ is a learnable node-scoring network that maps $\mathbb{R}^{d_x}$ to $\mathbb{R}$ and is parameterized by $\theta_{\rho}$.
Hence, the FC selects the \(L\) nodes:
\begin{equation}
	\mathcal{K}_u^{(L)}(n)
	\triangleq
	\operatorname{TopL}_{k\in\mathcal{K}_u}
	\bigl(\rho_{k,u}(n)\bigr),
	\label{eq:fc_topL_set}
\end{equation}
where \(L\in\mathbb{N}\) is a design parameter satisfying \(1\le L\le |\mathcal{K}_u|\).
To avoid ambiguity between node labels and matrix indices, we enumerate the selected nodes as \(\mathcal{K}_u^{(L)}(n)
=
\{\kappa_1,\kappa_2,\ldots,\kappa_L\}\).
A graph is then constructed over the selected nodes
$
\mathcal{G}_u^{(L)}(n)
\triangleq
\bigl(
\mathcal{V}_u^{(L)}(n),\,
\mathcal{E}_u^{(L)}(n)
\bigr)
$,
where $\mathcal{V}_u^{(L)}(n)$ and $\mathcal{E}_u^{(L)}(n)$ are the set of nodes and the set of edges of the graph, respectively. In this paper, we consider $\mathcal{V}_u^{(L)}(n)=\mathcal{K}_u^{(L)}(n)$.
The graph is complete on the selected nodes, i.e., \((\kappa_i,\kappa_j)\in\mathcal{E}_u^{(L)}(n)\), \(\forall\, i,j\in\{1,\ldots,L\},\; i\neq j\).

For the selected nodes, the corresponding node features are simply obtained
by evaluating \eqref{eq:fc_node_evidence} at \(k=\kappa_i\), i.e.,
\begin{equation}
	\mathbf{x}_{\kappa_i,u}(n)
	\triangleq
	\mathbf{x}_{k,u}(n)\big|_{k=\kappa_i},
	\qquad i=1,\ldots,L.
\end{equation}

For each pair of selected nodes \((\kappa_i,\kappa_j)\), the FC further defines the edge feature vector $\mathbf{e}_{i,j,u}(n) \in \mathbb{R}^{3}$ as following:
\begin{equation}
	\mathbf{e}_{i,j,u}(n) \triangleq 
	\bigl[ \Delta\tau_{i,j,u}(n),\; \Delta f_{i,j,u}(n),\; \Delta r_{i,j}(n) \bigr]^T,
\end{equation}
where
\(\Delta\tau_{i,j,u}(n)=\hat{\tau}_{\kappa_i,u}(n)-\hat{\tau}_{\kappa_j,u}(n)\), \(\Delta f_{i,j,u}(n)=\hat{f}_{\mathrm d}^{(\kappa_i,u)}(n)-\hat{f}_{\mathrm d}^{(\kappa_j,u)}(n)\), and \(\Delta r_{i,j}(n)=\|\mathbf{r}_{\kappa_i}(n)-\mathbf{r}_{\kappa_j}(n)\|\).

Using the pairwise delay and Doppler consistency, the FC constructs a weighted adjacency matrix
\(\mathbf{A}_u^{(L)}(n)\in\mathbb{R}^{L\times L}\), whose \((i,j)\)-th element is given by
\begin{align}
	&[\mathbf{A}_u^{(L)}(n)]_{i,j}
	\triangleq \nonumber \\
	&\quad \begin{cases}
		\exp\!\left(
		-\dfrac{
			\left|\Delta\tau_{i,j,u}(n)\right|
		}{\sigma_\mathrm{\tau}}
		-\dfrac{
			\left|\Delta f_{i,j,u}(n)\right|
		}{\sigma_\mathrm{d}}
		\right),
		& i\neq j,\\[3mm]
		0, & i=j,
	\end{cases}
	\label{eq:fc_adjacency}
\end{align}
where \(\sigma_\mathrm{\tau}>0\) and \(\sigma_\mathrm{d}>0\) are design parameters controlling the sensitivity to delay and Doppler mismatch, respectively.
Next, the FC stacks the selected node features into the matrix $\mathbf{X}_u^{(L)}(n) \in\mathbb{R}^{L\times d_x}$, defined as
\begin{equation}
	\mathbf{X}_u^{(L)}(n)
	\triangleq
	\begin{bmatrix}
		\mathbf{x}_{\kappa_1,u}^T(n)\\
		\mathbf{x}_{\kappa_2,u}^T(n)\\
		\vdots\\
		\mathbf{x}_{\kappa_L,u}^T(n)
	\end{bmatrix}.
	\label{eq:fc_stacked_X}
\end{equation}
These features are first mapped into a latent space through a trainable input embedding function \(\mathbf{H}_u^{(0)}(n)
\triangleq
\phi_{\mathrm{in}}\!\bigl(\mathbf{X}_u^{(L)}(n)\bigr)\),
where $\mathbf{H}_u^{(0)}(n)\in\mathbb{R}^{L\times d_h}$, \(\phi_{\mathrm{in}}:\mathbb{R}^{d_x}\rightarrow\mathbb{R}^{d_h}\) is applied row-wise, and \(d_h\in\mathbb{N}\) is the hidden embedding dimension.
Graph message passing is then performed over \(L_{\mathrm G}\) layers with
\begin{equation}
	\mathbf{H}_u^{(\ell+1)}(n)
	=
	\varphi^{(\ell)}
	\!\left(
	\widetilde{\mathbf{D}}_u^{-\frac{1}{2}}(n)
	\widetilde{\mathbf{A}}_u^{(L)}(n)
	\widetilde{\mathbf{D}}_u^{-\frac{1}{2}}(n)
	\mathbf{H}_u^{(\ell)}(n)\mathbf{W}^{(\ell)}
	\right),
	\label{eq:fc_gnn_layer}
\end{equation}
for \(\ell=0,1,\ldots,L_{\mathrm G}-1\), where \(\widetilde{\mathbf{A}}_u^{(L)}(n)
\triangleq
\mathbf{A}_u^{(L)}(n)+\mathbf{I}_L\), and
\(\mathbf{I}_L\) is the \(L\times L\) identity matrix,
\(\widetilde{\mathbf{D}}_u(n)\in\mathbb{R}^{L\times L}\) is the diagonal degree matrix associated with \(\widetilde{\mathbf{A}}_u^{(L)}(n)\),
\(\mathbf{W}^{(\ell)}\) is the trainable weight matrix of the \(\ell\)-th layer,
and \(\varphi^{(\ell)}(\cdot)\) is the corresponding nonlinear activation function.
The final graph-layer output is denoted by
\begin{equation}
	\mathbf{H}_u^{(L_{\mathrm G})}(n)
	=
	\begin{bmatrix}
		\mathbf{s}_{1,u}^T(n)\\
		\mathbf{s}_{2,u}^T(n)\\
		\vdots\\
		\mathbf{s}_{L,u}^T(n)
	\end{bmatrix},
	\label{eq:fc_final_node_embeddings}
\end{equation}
where \(\mathbf{s}_{i,u}(n)\in\mathbb{R}^{d_h}\) is the final embedding associated with the selected node \(\kappa_i\) and \(\mathbf{H}_u^{(L_{\mathrm G})}(n)
\in\mathbb{R}^{L\times d_h}\).

The FC aggregates the selected node embeddings via attention-based readout, where the attention logit of the \(i\)-th node is
\begin{equation}
	a_{i,u}(n)
	\triangleq
	\mathbf{v}_{\mathrm a}^T
	\tanh\,\bigl(
	\mathbf{W}_{\mathrm a}\mathbf{s}_{i,u}(n)
	\bigr)
	=\frac{
		\exp\,\bigl(a_{i,u}(n)\bigr)
	}{
		\sum_{j=1}^{L}\exp\,\bigl(a_{j,u}(n)\bigr)
	}
	,
	\label{eq:fc_attention_logit}
\end{equation}
where \(\mathbf{W}_{\mathrm a}\) and \(\mathbf{v}_{\mathrm a}\) are trainable attention parameters.
The graph-level fused representation is then obtained as
\begin{equation}
	\mathbf{F}_{\mathrm g}^{(u)}(n)
	\triangleq
	\sum_{i=1}^{L}
	\alpha_{i,u}(n)\mathbf{s}_{i,u}(n),
	\label{eq:fc_graph_representation}
\end{equation}
where \(\mathbf{F}_{\mathrm g}^{(u)}(n)\in\mathbb{R}^{d_h}\).
To preserve the geometric consistency information already introduced in
\eqref{eq:tri_score} and \eqref{eq:dist_score}, the FC augments the graph-level representation with the triangulation consistency score and the identity-location consistency score, which is represented by $
\bar{\mathbf{F}}_{\mathrm g}^{(u)}(n)\in\mathbb{R}^{d_h+2}$ and defined as 
\begin{equation}
	\bar{\mathbf{F}}_{\mathrm g}^{(u)}(n)
	\triangleq
	\Bigl[
	\bigl(\mathbf{F}_{\mathrm g}^{(u)}(n)\bigr)^T,\,
	f_{\mathrm{tri},u}(n),\,
	f_{\mathrm{dist},u,u^\star}(n)
	\Bigr]^T.
	\label{eq:fc_augmented_graph_representation}
\end{equation}

The final authentication score generated by the FC is \(S_{\mathrm{FC}}^{(u)}(n)
\triangleq
s_{\theta_{\mathrm{FC}}}\!\bigl(
\bar{\mathbf{F}}_{\mathrm g}^{(u)}(n)
\bigr)\),
where \(s_{\theta_{\mathrm{FC}}}\) is decision neural network.
Here, lower values of \(S_{\mathrm{FC}}^{(u)}(n)\) indicate stronger evidence of
legitimacy, whereas larger values indicate stronger evidence for rejection.
The overall FC-side parameter set is
\begin{equation}
	\theta_{\mathrm{FC}}
	\triangleq
	\left\{
	\theta_{\rho},\,
	\phi_{\mathrm{in}},\,
	\{\mathbf{W}^{(\ell)}\}_{\ell=0}^{L_{\mathrm G}-1},\,
	\mathbf{W}_{\mathrm a},\,
	\mathbf{v}_{\mathrm a},\,
	\mathbf{w}_{\mathrm o},\,
	b_{\mathrm o}
	\right\},
\end{equation}
where \(\mathbf{w}_{\mathrm o}\in\mathbb{R}^{d_h+2}\) and \(b_{\mathrm o}\in\mathbb{R}\).
Finally, the FC makes the authentication decision by
\begin{equation}
	D_{\mathrm{final}}^{(u)}(n)
	=
	\mathbf{1}\!\left\{
	S_{\mathrm{FC}}^{(u)}(n)>\tau_{\mathrm{FC}}
	\right\},
	\label{eq:fc_final_decision}
\end{equation}
where \(\tau_{\mathrm{FC}}\in(0,1)\) is the decision threshold. Moreover, \(D_{\mathrm{final}}^{(u)}(n)=0\) denotes acceptance as a legitimate transmission and \(D_{\mathrm{final}}^{(u)}(n)=1\) denotes rejection.

The proposed fusion strategy balances scalability and expressiveness by limiting graph construction and message passing to the \(L\) most reliable nodes, thereby reducing computational cost when \(|\mathcal{K}_u|\) is large. Meanwhile, the graph captures pairwise consistency in delay, Doppler, and geometry, enabling reliable authentication in dynamic NTN environments.

\subsection{Verification Score and ROC Metrics}

The receiver operating characteristic (ROC) curve shows the tradeoff between correctly accepting legitimate transmissions and wrongly accepting attacker transmissions as the threshold \(\tau_{\mathrm{FC}}\) changes. Here, the positive class corresponds to legitimate users. For a given \(\tau_{\mathrm{FC}}\), the true positive rate (TPR) and false positive rate (FPR) are defined as
\begin{align}
	\mathrm{TPR}(\tau_{\mathrm{FC}})
	&=
	\frac{
		TP(\tau_{\mathrm{FC}})
	}{
		TP(\tau_{\mathrm{FC}})+FN(\tau_{\mathrm{FC}})
	},
	\label{eq:tpr}
	\\
	\mathrm{FPR}(\tau_{\mathrm{FC}})
	&=
	\frac{
		FP(\tau_{\mathrm{FC}})
	}{
		FP(\tau_{\mathrm{FC}})+TN(\tau_{\mathrm{FC}})
	}.
	\label{eq:fpr}
\end{align}
Here, \(TP(\tau_{\mathrm{FC}})\) is the number of legitimate transmissions correctly accepted \((D_{\mathrm{final}}^{(u)}(n)=0)\), \(FN(\tau_{\mathrm{FC}})\) is the number of legitimate transmissions wrongly rejected \((D_{\mathrm{final}}^{(u)}(n)=1)\), \(FP(\tau_{\mathrm{FC}})\) is the number of attacker transmissions wrongly accepted \((D_{\mathrm{final}}^{(u)}(n)=0)\), and \(TN(\tau_{\mathrm{FC}})\) is the number of attacker transmissions correctly rejected \((D_{\mathrm{final}}^{(u)}(n)=1)\).
The area under the ROC curve (AUC) is defined as
\begin{align}
	\mathrm{AUC}
	=&
	\int_{0}^{1}
	\mathrm{TPR}(\mathrm{FPR})\,d\mathrm{FPR} \nonumber \\
	=&
	\int
	\mathrm{TPR}(\tau_{\mathrm{FC}})
	\left|
	\frac{d\,\mathrm{FPR}(\tau_{\mathrm{FC}})}
	{d\tau_{\mathrm{FC}}}
	\right|
	d\tau_{\mathrm{FC}},
	\label{eq:auc_standard}
\end{align}
where the absolute value ensures a nonnegative area under the ROC curve.
Under the convention in \eqref{eq:fc_final_decision}, increasing
\(\tau_{\mathrm{FC}}\) generally makes acceptance less strict, since a
transmission is accepted when
\(S_{\mathrm{FC}}^{(u)}(n)\leq \tau_{\mathrm{FC}}\).

\subsection{Summary of the Algorithms}
\label{sec:solution}

As shown in Fig.~\ref{fig:flow_chart}, the proposed framework consists of four stages: data collection, causal estimation, meta-learning, and authentication. First, NTN observations are collected across environments and nodes, including variables, treatments, measurements, features, and labels. Second, the causal module estimates interventional quantities such as DR means, ACE, and CATE. Third, the meta-learning module optimizes $\phi$ using MAML with causal-invariance and DR-consistency regularization to improve cross-environment generalization. Finally, the trained model produces the FC authentication score and accept/reject decision, while the graph-based attack detection module is activated if the claimed identity is not legitimate.

\begin{figure}[t]
	\centering
	\includegraphics[width=0.98\linewidth]{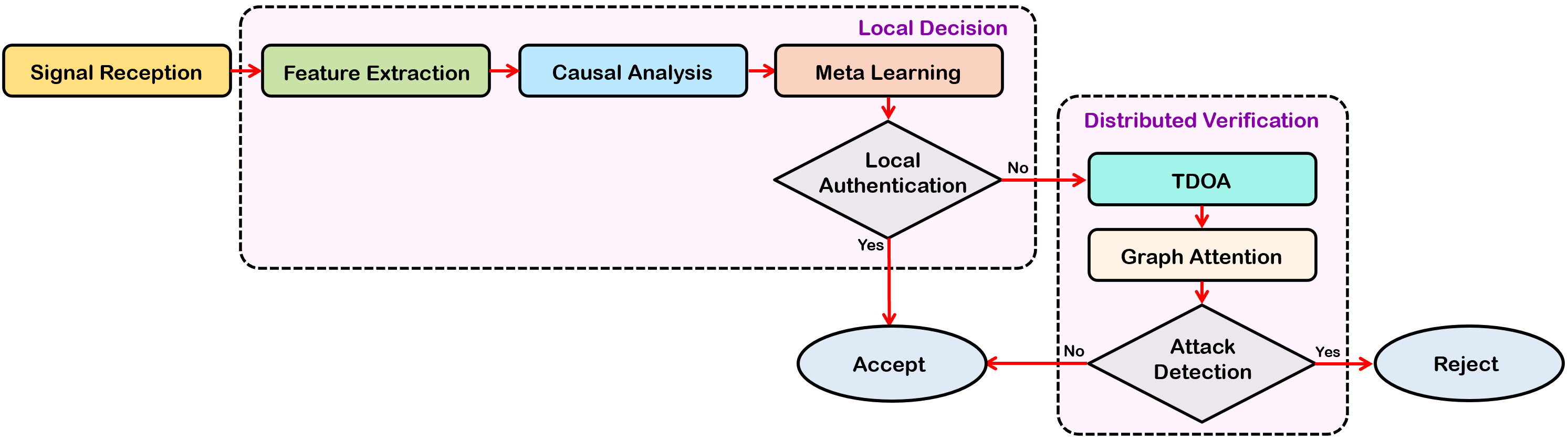}
	\caption{Summary of the proposed method and algorithms.}
	\label{fig:flow_chart}
\end{figure}

Algorithms~\ref{alg:causal} and \ref{alg:metalearning} summarize the causal estimation and meta-learning procedures, respectively. The causal estimates guide training through invariance and DR consistency regularizers, encouraging the model to exploit stable causal relationships rather than environment-specific correlations. The final model is evaluated on unseen NTN environments.

\begin{algorithm}[t]
	\caption{Causal Effect Estimation}
	\label{alg:causal}
	\footnotesize
	\begin{algorithmic}[1]
		\Require
		\(\mathcal{T}\),
		\(
		\mathcal{D}_{\mathrm{causal}}
		=
		\left\{
		\left(
		Z_{k,u}^{(i)},
		T^{(i)},
		\mathbf{f}_{k,u}^{(i)},
		Y_{k,u}^{(i)}
		\right)
		\right\}_{i=1}^{N_s}
		\)
		\Ensure
		\(\widehat{\mu}^{\mathrm{DR}}(t)\),
		\(\widehat{\mathrm{ACE}}(t_1,t_0)\), and
		\(\widehat{\mathrm{CATE}}(t_1,t_0\mid z_{k,u})\)
		
		\State Estimate propensity scores \(e_t(z)\) via \eqref{eq:propensity}, \(\forall t\in\mathcal{T}\).
		\State Estimate outcome models \(m_t(z)\) via \eqref{eq:outcome}, \(\forall t\in\mathcal{T}\).
		\State Compute DR means \(\widehat{\mu}^{\mathrm{DR}}(t)\) via \eqref{eq:dr}, \(\forall t\in\mathcal{T}\).
		\State Compute ACEs \(\widehat{\mathrm{ACE}}(t_1,t_0)\) via \eqref{eq:ace}, \(\forall (t_1,t_0)\in\mathcal{T}^2\).
		\State Compute \(\widehat{\mathrm{CATE}}(t_1,t_0\mid z_{k,u})\) via \eqref{eq:cate}, 
		\(\forall (t_1,t_0)\in\mathcal{T}^2\) and \(z_{k,u}\).
		\State \Return
		\(\{\widehat{\mu}^{\mathrm{DR}}(t)\}_{t\in\mathcal{T}}\),
		\(\{\widehat{\mathrm{ACE}}(t_1,t_0)\}_{t_0,t_1\in\mathcal{T}}\), \\ \quad and
		\(\{\widehat{\mathrm{CATE}}(t_1,t_0\mid z_{k,u})\}_{t_0,t_1\in\mathcal{T},\, z_{k,u}\in Z_{k,u}}\).
	\end{algorithmic}
\end{algorithm}

\begin{algorithm}[t]
	\caption{Meta-Learning with Causal Invariance}
	\label{alg:metalearning}
	\footnotesize
	\begin{algorithmic}[1]
		\Require
		\(\{D_{\mathrm{train}}^{e,k},D_{\mathrm{val}}^{e,k}\}\),
		\(\alpha,\beta\),
		\(\lambda_{\mathrm{IRM}},\lambda_{\mathrm{DR}}\),
		and outputs from Algorithm~\ref{alg:causal}
		\Ensure Learned initialization \(\phi^\star\)
		
		\State Initialize \(\phi\) randomly.
		
		\While{not converged}
		\State Compute task-adapted parameters \(\phi'_{e,k}\) using \eqref{eq:maml_inner}, 
		\(\forall e\in\mathcal{E}, k\in\mathcal{K}\).
		
		\State Evaluate validation losses 
		\(\mathcal{L}_{\mathrm{task}}^{e,k}(\phi'_{e,k};D_{\mathrm{val}}^{e,k})\), 
		\(\forall e\in\mathcal{E}, k\in\mathcal{K}\).
		
		\State Compute \(\mathcal{R}_{\mathrm{IRM}}(\phi)\) and 
		\(\mathcal{R}_{\mathrm{DR}}(\phi)\) using \eqref{eq:irm} and \eqref{eq:dr_reg}.
		
		\State Update \(\phi\) using the outer update in \eqref{eq:maml_outer}.
		\EndWhile
		
		\State \Return \(\phi^\star \leftarrow \phi\).
	\end{algorithmic}
\end{algorithm}

\section{Computational Complexity Analysis}
\label{sec:complexity}

We analyze the computational complexity of the proposed framework for both offline training and online inference. 

For each node \(k\), feature extraction in \eqref{eq:f1_new}--\eqref{eq:f8_new} is dominated by FFT operations and subspace projection, resulting in per-sample complexity
\(
\mathcal{O}(M_k \log M_k + M_k r_k)
\). The local Mahalanobis scoring has negligible complexity compared with feature extraction. Hence, the total local-processing complexity over \(N_s\) samples is \(\mathcal{O}\!\left(
N_s
\sum_{k=1}^{|\mathcal{K}|}
(M_k \log M_k + M_k r_k)
\right)\).

Let \(M \triangleq \max_k M_k\) denote the maximum antenna number.
Under practical LoS-dominant NTN conditions with \(r_k=\mathcal{O}(1)\), this reduces to
\(
\mathcal{O}(N_s |\mathcal{K}| M \log M)
\),
while in rich scattering it becomes
\(
\mathcal{O}(N_s |\mathcal{K}| M^2)
\).

The causal estimation module fits the propensity and outcome models over a dataset of size \(N_s\). Since \(|\mathcal{T}|\) is fixed and small, the overall causal-estimation complexity scales as
\(
\mathcal{O}(N_s |\mathcal{T}|)
\).
For meta-learning, let \(C_{\mathrm{ML}}\) denote the complexity of one forward-backward pass, where
\(
C_{\mathrm{ML}}=\mathcal{O}(d_x d_h + d_h^2)
\).
The total training complexity over \(n_{\mathrm{iter}}\) outer iterations becomes \(\mathcal{O}\!\left(
n_{\mathrm{iter}} N_s
(C_{\mathrm{ML}}+|\mathcal{T}|)
\right)\).
If the feature extractor is trained end-to-end, the local feature-extraction cost is additionally included.
Since \(d_x\), \(d_h\), and \(|\mathcal{T}|\) remain bounded in practical implementations, the meta-learning and causal-estimation terms scale approximately linearly with \(n_{\mathrm{iter}} N_s\).

For one authentication attempt, each node performs channel estimation, feature extraction, and local scoring with complexity
\(
\mathcal{O}(M_k \log M_k + M_k r_k)
\).
If distributed verification is activated, Top-\(L\) node selection requires
\(
\mathcal{O}(|\mathcal{K}_u| \log |\mathcal{K}_u|)
\),
while graph construction and message passing require
\(
\mathcal{O}(L^2)
\)
and
\(
\mathcal{O}(L_{\mathrm G} L^2 d_h)
\),
respectively. Optional TDOA localization adds
\(
\mathcal{O}(L^3)
\).
However, since \(L\), \(L_{\mathrm G}\), and \(d_h\) are fixed small values in our implementation, and TDOA is only triggered when local authentication fails, the dominant online complexity is \(\mathcal{O}\!\left(
\sum_{k \in \mathcal{K}_u}
(M_k \log M_k + M_k r_k)
\right)\).
For homogeneous node dimensions \(M_k \approx M\) and bounded \(r_k\), this reduces to
\(
\mathcal{O}(|\mathcal{K}_u| M \log M)
\).

The local-only baseline employs a conventional deep learning-based authentication model using the same local feature extraction pipeline, but without causal estimation, meta-learning adaptation, graph-based fusion, or TDOA-based attack detection. It still performs local feature extraction and local scoring. Thus, its offline complexity is approximately
\(
\mathcal{O}(N_s |\mathcal{K}| M \log M)
\),
while its online complexity is
\(
\mathcal{O}(|\mathcal{K}_u| M \log M)
\).
Under LoS-dominant NTN conditions, the proposed framework has simplified offline complexity \(\mathcal{O}\!\left(
N_s |\mathcal{K}| M \log M
+
n_{\mathrm{iter}} N_s
\right)\),
and dominant online complexity
\(
\mathcal{O}(|\mathcal{K}_u| M \log M)
\).
Thus, the proposed framework introduces additional offline computation for causal estimation and meta-learning, while the online cost is dominated by local processing and remains scalable.

Table~\ref{tab:complexity_summary} summarizes the dominant computational complexity comparison. The online expressions reported in the table correspond to the dominant terms after treating \(L\), \(L_{\mathrm G}\), \(d_h\), and conditional TDOA localization as lower-order contributions.
\begin{table}[t]
	\centering
	\caption{Final computational complexity comparison.}
	\label{tab:complexity_summary}
	\scriptsize
	\setlength{\tabcolsep}{3pt}
	\renewcommand{\arraystretch}{1.2}
	\begin{tabular}{ll}
		\hline
		\textbf{Scheme} & \textbf{Complexity} \\
		\hline
		
		\textbf{Proposed}
		& \textbf{Offline:} \(\mathcal{O}\!\left(N_s |\mathcal{K}| M \log M + n_{\mathrm{iter}} N_s\right)\) \\
		& \textbf{Online:} \(\mathcal{O}\!\left(|\mathcal{K}_u| M \log M\right)\) \\
		
		\hline
		
		\textbf{Local-only}
		& \textbf{Offline:} \(\mathcal{O}\!\left(N_s |\mathcal{K}| M \log M\right)\) \\
		& \textbf{Online:} \(\mathcal{O}\!\left(|\mathcal{K}_u| M \log M\right)\) \\
		
		\hline
		
		\textbf{Typical parameter ranges}
		& \begin{tabular}[t]{l}
			\(N_s \approx 10^3\!-\!10^4,\; |\mathcal{K}| \approx 2\!-\!20\) \\
			\(|\mathcal{K}_u| \approx 2\!-\!10,\; M \approx 32\!-\!256\) \\
			\(n_{\mathrm{iter}} \approx 10\!-\!50\)
		\end{tabular} \\
		
		\hline
	\end{tabular}
\end{table}

\section{Simulation Results}
\label{sec:results}

We first generate a multi-environment dataset and train the model. Then, we test it under unseen conditions and analyze the features of the proposed multi-feature fingerprint.

\subsection{Simulation Setup}

\subsubsection{Data Generation}

We generate a dataset covering HAPS and satellite scenarios. Six scenarios are considered: five for training and one for testing. Each scenario includes ten users with IDs \(1\) to \(10\), and each user performs ten movements. The objective is user-ID recognition from the extracted fingerprint features.
The uplink channel follows a Rician fading model with carrier frequency \(f_c=2\)~GHz. For each task, the environmental variables are sampled from the ranges in Table~\ref{tab:env_sweep}.
We generate \(10\) samples per user per scenario. Thus, the first-phase training set contains \(500\) samples, i.e., \(5\) scenarios \(\times\) \(10\) users \(\times\) \(10\) samples, while the test set contains \(100\) samples, i.e., \(1\) scenario \(\times\) \(10\) users \(\times\) \(10\) samples. The physical feature vector \(\mathbf{f}_{k,u}(n)\) is computed using \eqref{eq:f1_new}--\eqref{eq:f8_new}.
Channel estimation includes AWGN with power \(\sigma_k^2=10^{-12}\)~W, TOA noise with standard deviation \(\sigma_{\mathrm{toa}}=5\)~ns, and Doppler noise with standard deviation \(0.5\)~Hz. Legitimate users and attackers are located within a \(50\times50\)~km ground area centered at the origin.
For each environment, \(N_{\mathrm{enr}}=100\) enrollment snapshots per legitimate user are collected to estimate \(\boldsymbol{\mu}_h^{(k,u)}\), \(\mathbf{R}_h^{(k,u)}\), and \(\mathbf{P}_{\mathrm{ref}}^{(k,u)}\). The reference subspace uses the two dominant eigenvectors, i.e., \(r_k=2\). The reference Doppler and delay values are denoted by \(f_{\mathrm{d,ref}}^{(k,u)}\) and \(\tau_{k,u}^{\mathrm{(ref)}}\), respectively.

\subsubsection{Causal Meta-Learning for Local Authentication}

The causal meta-learning framework in Algorithms~\ref{alg:causal} and~\ref{alg:metalearning} is trained on a set of \(|\mathcal{E}| = 5\) distinct environments, corresponding to the scenarios in the data generation code. Each environment contains \(10\) users (either legitimate or attacker). The feature dimension is \(d_f = 8\). For each task, we randomly sample \(5\) support samples and \(5\) query samples per user in a few-shot setting. The inner learning rate is \(\alpha = 0.01\) with \(5\) inner gradient steps. The outer learning rate is \(\beta = 0.001\). The meta-learner is trained for \(n_{\mathrm{iter}} = 160\) outer epochs, with \(20\) episodes per epoch. 
The hyperparameters \(\lambda_{\mathrm{IRM}}\) and \(\lambda_{\mathrm{DR}}\) are set to \(0.1\) and \(0.05\), respectively. 
The neural network \(\Phi_{\phi}\), used as the feature extractor is a \(4\)-layer MLP with hidden size \(64\) and output dimension \(d_f\), one output per feature. 

\subsubsection{Distributed Verification and Attack Detection}

For the second-phase experiments, i.e., distributed verification using TDOA and location-consistency scores, we consider \(|\mathcal{K}| = 4\) nodes located at \((\pm50,\pm50)\)~km horizontally. Each node is placed at altitude \(20\)~km, has \(M_k = 16\) antennas, and moves with velocity \(v_k=15\)~m/s, as used in the verification code. The legitimate user and the attacker are placed on the ground within the \((\pm50,\pm50)\)~km region.
The training set consists of \(160\) samples for legitimate users and \(160\) samples for attackers. The test set contains \(100\) samples for each class. 
The TDOA estimation uses \(10\) consecutive time steps. The position-stability score uses \(\sigma_{\mathrm{pos}} = 5\)~m, and the distance-consistency score uses \(\sigma_{\mathrm{dist}} = 500\)~m. 

The FC classifier is a two-layer MLP with hidden units \((64,32)\) and a sigmoid output. It is trained with the Adam optimizer for \(1000\) epochs, without early stopping. The final decision threshold is \(\tau_{\mathrm{FC}} = 0.5\). 
All results are averaged over \(n_{\mathrm{MC}} = 10\) independent Monte Carlo runs. The SNR sweep varies the transmit power from \(10^{-4}\) to \(10^{2}\)~W, corresponding to an SNR range from \(-10\) to \(30\)~dB.

\subsubsection{Test Scenario}

In the final test, the system is evaluated on unseen environment combinations. The meta-learner is adapted to the new environment using \(10\) support samples per user via the inner update and tested on \(100\) query samples. In distributed verification, the FC collects local soft scores \(p_{k,u,u_k^\star}(n)\) and TDOA measurements \(\hat{\tau}_{k,u}(n)\) from all nodes. The reported metrics are AUC and authentication accuracy.

\begin{table}[t]
	\centering
	\caption{Ranges of environment variables for task generation.}
	\label{tab:env_sweep}
	\begin{tabular}{@{}l c l@{}}
		\toprule
		\textbf{Parameter} & \textbf{Symbol} & \textbf{Range / Values} \\
		\midrule
		Number of antennas & \(M_k\) & \(\{16,\; 64,\; 128,\; 256\}\) \\
		SNR at node \(k\) & \(\gamma_{k,u}\) & \(\{-10,\,-5,\dots,30\}\)~dB \\
		Angular spread & \(\sigma_{\theta,k}\) & \(\{2^\circ,\;5^\circ,\;10^\circ\}\) \\
		Rician factor & \(K_{\mathrm r}^{(k,u)}\) & \(\{5,\;10,\;20\}\) \\
		ADC resolution & \(N_{\mathrm{bit}}^{(k)}\) & \(\{4,\;6,\;8\}\) bits \\
		Pilot reuse groups & \(K_{\mathrm p}^{(k)}\) & \(\{2,\;3,\;4\}\) \\
		Elevation angle & \(\theta_{\mathrm{elev}}^{(k,u)}\) & \(\{30^\circ,\;60^\circ,\;90^\circ\}\) \\
		\addlinespace
		\multicolumn{3}{c}{\textbf{HAPS-specific}} \\
		HAPS altitude & \(H_{\mathrm{alt}}^{(k)}\) & \(\{20,\;25\}\) km \\
		HAPS velocity & \(v_k\) & \(\{15,\;20\}\) m/s \\
		\addlinespace
		\multicolumn{3}{c}{\textbf{LEO/GEO-specific}} \\
		Satellite altitude & \(H_{\mathrm{alt}}^{(k)}\) & \(\{500,\;1200,\;2000,\;35786\}\) km \\
		Satellite velocity & \(v_k\) & \(\{7200,\;7500,\;7800,\;3070\}\) m/s \\
		\bottomrule
	\end{tabular}
\end{table}

\subsection{Multi-Feature Comparison and Feature Importance}

Fig.~\ref{fig:roc_features} compares the AUC of different methods. The proposed method achieves the highest AUC by using all eight features, learning the causal relationships, and adapting to new environments via meta-learning. In contrast, the deep-learning baseline may overfit environment-specific patterns and perform worse in unseen NTN conditions.

\begin{figure}[t]
	\centering
	\includegraphics[width=0.8\columnwidth]{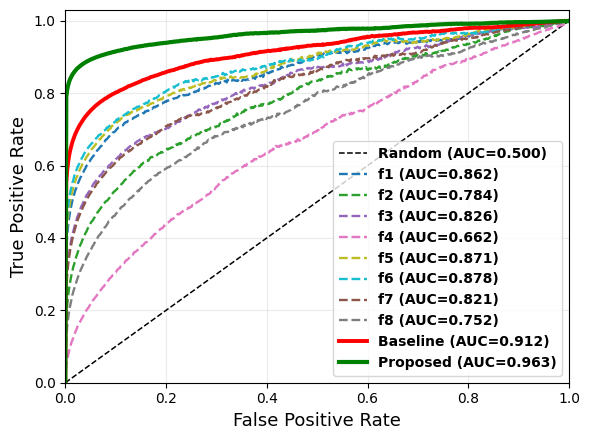}
	\caption{Average ROC curves for single-feature classifiers, baseline fusion, and the proposed all-feature fusion method across the considered NTN scenarios.}
	\label{fig:roc_features}
\end{figure}

The relevance of PLA features depends on the NTN scenario, link geometry, interference level, and channel estimation quality. Overall, geometry- and spatial-domain features, i.e., \(f_{k,u,1}\), \(f_{k,u,5}\), and \(f_{k,u,6}\), are generally reliable across most NTN deployments, whereas interference-aware features, such as \(f_{k,u,7}\) and \(f_{k,u,8}\), become increasingly important in dense multi-user settings with IUI or pilot contamination.

Feature importance varies with propagation conditions, platform type, system parameters, and attacker capability. In sparse LoS-dominant scenarios, angular, subspace, and geometry-dependent features, i.e., \(f_{k,u,1}\), \(f_{k,u,2}\), \(f_{k,u,5}\), and \(f_{k,u,6}\), are typically the most discriminative, whereas removing \(f_{k,u,7}\) or \(f_{k,u,8}\) causes only limited degradation due to weak interference. In dense NTN access scenarios, post-combining and interference-aware features become more important; hence, \(f_{k,u,7}\), \(f_{k,u,8}\), \(f_{k,u,3}\), and \(f_{k,u,5}\) are more influential, while the snapshot-dependent projection feature \(f_{k,u,4}\) is usually less robust under noise and interference.

The dominant features also depend on the NTN platform. For LEO links, high mobility increases the importance of the delay--Doppler feature \(f_{k,u,6}\), together with spatial features \(f_{k,u,1}\), \(f_{k,u,3}\), and \(f_{k,u,5}\). For GEO links, the weaker Doppler effect makes \(f_{k,u,1}\), \(f_{k,u,2}\), and \(f_{k,u,5}\) dominant, although the delay component of \(f_{k,u,6}\) can still be useful. For HAPS links, stable geometry and shorter propagation distance make angular and subspace features, particularly \(f_{k,u,1}\), \(f_{k,u,2}\), and \(f_{k,u,5}\), highly informative.

System conditions further affect the ranking. High SNR and large antenna arrays improve channel estimation and angular resolution, strengthening \(f_{k,u,1}\), \(f_{k,u,2}\), and \(f_{k,u,5}\). At low SNR, with few antennas, strong NLoS propagation, large angular spread, or pilot contamination, instantaneous and beam-domain features become less stable, while \(f_{k,u,3}\), \(f_{k,u,5}\), \(f_{k,u,6}\), \(f_{k,u,7}\), and \(f_{k,u,8}\) provide more robust evidence.

Attacker location and capability also influence feature discrimination. When the attacker is far from the legitimate user, \(f_{k,u,6}\) is highly effective due to distinct geometric signatures. However, for a nearby or more capable attacker that partially mimics delay--Doppler characteristics, authentication relies more on spatial consistency, combiner similarity, subspace deviation, and post-combining SINR, represented by \(f_{k,u,1}\), \(f_{k,u,3}\), \(f_{k,u,5}\), and \(f_{k,u,7}\), respectively.

Overall, a qualitative ranking for typical NTN scenarios is
\begin{align}
	&f_{k,u,6}
	\approx
	f_{k,u,5}
	\approx
	f_{k,u,1}
	>
	f_{k,u,3}
	\approx
	f_{k,u,7}
	>
	f_{k,u,2} \nonumber \\
	&\quad>
	f_{k,u,8}
	>
	f_{k,u,4}.
	\label{eq:general_feature_ranking}
\end{align}
In dense and interference-limited scenarios, the ranking shifts toward interference-aware and combiner-based features as
\begin{align}
	&f_{k,u,7}
	\approx
	f_{k,u,8}
	\approx
	f_{k,u,3}
	>
	f_{k,u,5}
	\approx
	f_{k,u,6}
	\approx
	f_{k,u,1} \nonumber \\
	&\quad>
	f_{k,u,2}
	>
	f_{k,u,4}.
	\label{eq:dense_feature_ranking}
\end{align}
In sparse LoS-dominant environments, the ranking becomes
\begin{align}
	&f_{k,u,1}
	\approx
	f_{k,u,5}
	\approx
	f_{k,u,6}
	\approx
	f_{k,u,2}
	>
	f_{k,u,3}
	>
	f_{k,u,7} \nonumber \\
	&\quad>
	f_{k,u,8}
	\approx
	f_{k,u,4}.
	\label{eq:sparse_feature_ranking}
\end{align}

\begin{table*}[!t]
	\centering
	\caption{Qualitative importance of the considered physical-layer authentication features under different NTN environments.}
	\label{tab:feature_importance_ntn}
	\footnotesize
	\renewcommand{\arraystretch}{1.05}
	\setlength{\tabcolsep}{3.6pt}
	\begin{tabular}{|p{2.3cm}|p{1.4cm}|p{1.35cm}|p{1.4cm}|p{1.35cm}|p{1.55cm}|p{6.5cm}|}
		\hline
		\textbf{Feature Name (Notation)} & \textbf{Typical NTN} & \textbf{Dense Access} & \textbf{Sparse LoS} & \textbf{LEO} & \textbf{GEO/HAPS} & \textbf{Main Interpretation} \\
		\hline
		Channel Similarity ($f_{k,u,1}$) & High & Medium--High & Very High & High & Very High & Captures spatial consistency with the enrolled user. \\
		\hline
		Main-Lobe Energy ($f_{k,u,2}$) & Medium & Low--Medium & High & Medium & High & Useful when dominant angular support is stable; weaker under scattering, beam spread, or contamination. \\
		\hline
		Combiner Similarity ($f_{k,u,3}$) & Medium--High & High & Medium & High & Medium & More important when spatial filtering consistency matters, especially in dense interference-limited settings. \\
		\hline
		Projection Score ($f_{k,u,4}$) & Low & Low & Low--Medium & Low & Low & Usually the least robust feature due to sensitivity to noise, interference, and snapshot variability. \\
		\hline
		Subspace Deviation ($f_{k,u,5}$) & Very High & High & Very High & High & Very High & One of the most informative features, reflecting deviation from the enrolled spatial structure. \\
		\hline
		Delay--Doppler ($f_{k,u,6}$) & Very High & High & High & Very High & Medium & Strong when geometry and motion are discriminative; particularly important for LEO links. \\
		\hline
		SINR ($f_{k,u,7}$) & High & Very High & Low--Medium & High & Medium & Critical in dense access and strong-IUI conditions because it reflects user separability after combining. \\
		\hline
		INR ($f_{k,u,8}$) & Medium & Very High & Low & Medium & Low--Medium & Mainly useful in interference-heavy scenarios where illegitimate activity alters the interference level. \\
		\hline
	\end{tabular}
\end{table*}

\begin{table*}[!t]
	\centering
	\caption{Scenario-dependent qualitative feature importance for physical-layer authentication in NTN environments.}
	\label{tab:scenario_feature_ranking}
	\footnotesize
	\renewcommand{\arraystretch}{1.05}
	\setlength{\tabcolsep}{3.4pt}
	\begin{tabular}{|p{3.1cm}|p{1.75cm}|p{2.55cm}|p{1.80cm}|p{1.80cm}|p{5.3cm}|}
		\hline
		\textbf{Scenario} & \textbf{Primary} & \textbf{Secondary} & \textbf{Less} & \textbf{Least} & \textbf{Notes} \\
		\hline
		Typical NTN Environment & $f_{k,u,6}$, $f_{k,u,5}$, $f_{k,u,1}$ & $f_{k,u,3}$, $f_{k,u,7}$ & $f_{k,u,2}$, $f_{k,u,8}$ & $f_{k,u,4}$ & Geometry consistency and spatial structure are usually the strongest authentication cues. \\
		\hline
		Dense Multi-User Access & $f_{k,u,7}$, $f_{k,u,8}$, $f_{k,u,3}$ & $f_{k,u,5}$, $f_{k,u,6}$, $f_{k,u,1}$ & $f_{k,u,2}$ & $f_{k,u,4}$ & Interference-aware features dominate under strong IUI and pilot contamination.
		\\
		\hline
		Sparse LoS-Dominant & $f_{k,u,1}$, $f_{k,u,5}$, $f_{k,u,6}$, $f_{k,u,2}$ & $f_{k,u,3}$ & $f_{k,u,7}$ & $f_{k,u,4}$, $f_{k,u,8}$ & Spatial, angular, and geometric consistency are highly discriminative under weak interference.
		\\
		\hline
		LEO Links & $f_{k,u,6}$, $f_{k,u,1}$ & $f_{k,u,3}$, $f_{k,u,5}$ & $f_{k,u,7}$, $f_{k,u,2}$ & $f_{k,u,4}$ & High mobility enhances the value of delay--Doppler cues. \\
		\hline
		GEO/HAPS Links & $f_{k,u,1}$, $f_{k,u,2}$, $f_{k,u,5}$ & $f_{k,u,6}$, $f_{k,u,3}$ & $f_{k,u,7}$, $f_{k,u,8}$ & $f_{k,u,4}$ & Stable geometry reduces Doppler relevance and emphasizes spatial consistency. \\
		\hline
		Low-SNR & $f_{k,u,6}$, $f_{k,u,7}$ & $f_{k,u,5}$, $f_{k,u,3}$ & $f_{k,u,1}$, $f_{k,u,2}$ & $f_{k,u,4}$ & Spatial features degrade, while geometric and interference-aware cues become more useful. \\
		\hline
	\end{tabular}
\end{table*}

\begin{table}[!t]
	\centering
	\caption{Categorization of the considered physical-layer features.}
	\label{tab:feature_categories}
	\footnotesize
	\renewcommand{\arraystretch}{1.05}
	\setlength{\tabcolsep}{4.2pt}
	\begin{tabular}{p{4.5cm} p{3.0cm}}
		\toprule
		\textbf{Category} & \textbf{Feature Notation(s)} \\
		\midrule
		Spatial / Angular Consistency & $f_{k,u,1}$, $f_{k,u,2}$ \\
		Combiner / Receiver-Side Consistency & $f_{k,u,3}$, $f_{k,u,4}$ \\
		Subspace Structure Deviation & $f_{k,u,5}$ \\
		Geometry / Kinematics & $f_{k,u,6}$ \\
		Interference-Aware & $f_{k,u,7}$, $f_{k,u,8}$ \\
		\bottomrule
	\end{tabular}
\end{table}

These rankings complement empirical feature-importance analysis and show feature roles under different NTN conditions. Tables~\ref{tab:feature_importance_ntn} and~\ref{tab:scenario_feature_ranking} summarize them from feature- and scenario-centric perspectives.
Table~\ref{tab:feature_categories} groups the features into spatial, combiner-based, subspace-based, geometry-based, and interference-aware categories for easier interpretation.

Some redundancy exists among the proposed features. Specifically, \(f_{k,u,1}\), \(f_{k,u,2}\), and \(f_{k,u,5}\) all characterize spatial consistency and may therefore be correlated, particularly in strong LoS conditions with small angular spread. Similarly, \(f_{k,u,7}\) and \(f_{k,u,8}\) are both interference-related; however, \(f_{k,u,7}\) is usually more informative because it reflects the effective SINR, whereas \(f_{k,u,8}\) only measures interference relative to noise and is mainly useful in dense access scenarios.

To quantitatively validate the above observations, we perform a leave-one-feature-out ablation study.
Let \(\mathrm{AUC}_{\mathrm{all}}\) denote the AUC obtained using the complete features and let \(\mathrm{AUC}_{-f_i}\) denote the AUC obtained after removing the \(i\)-th feature. The contribution of feature \(f_i\) is derived as \(\Delta \mathrm{AUC}_i
=
\mathrm{AUC}_{\mathrm{all}}
-
\mathrm{AUC}_{-f_i}\).
A large value of \(\Delta \mathrm{AUC}_i\) indicates that the removed feature is important for authentication, whereas a value close to zero indicates that the feature has limited marginal contribution under the considered environment.

\begin{figure}[t]
	\centering
	\captionsetup[subfloat]{font=footnotesize}
	
	\subfloat[]{
		\includegraphics[width=0.46\columnwidth]{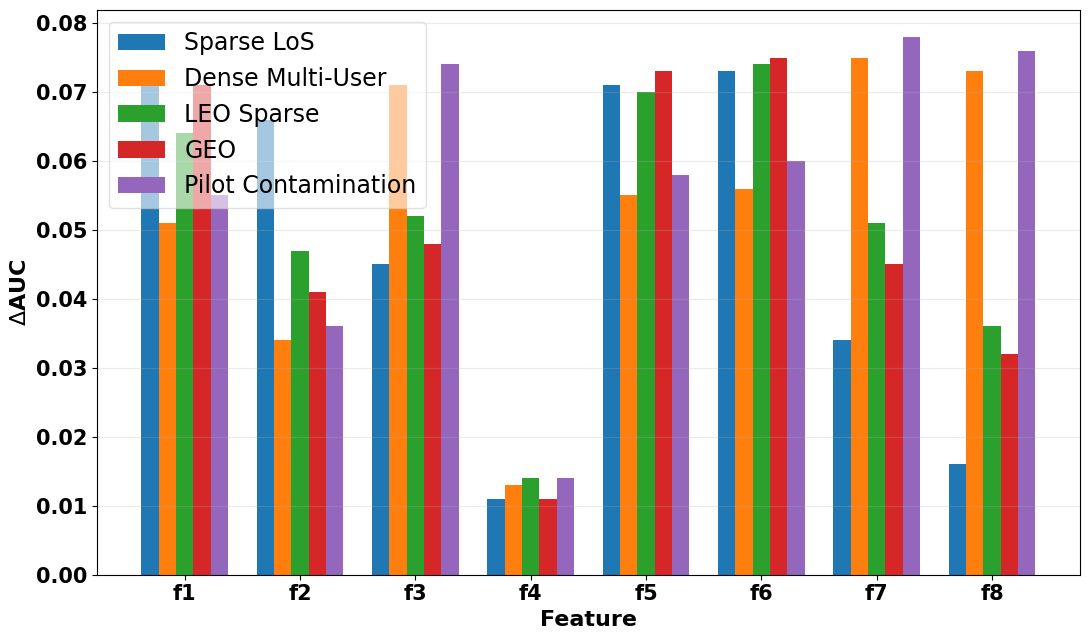}
		\label{fig:mil_auc}
	}
	\hfill
	\subfloat[]{
		\includegraphics[width=0.46\columnwidth]{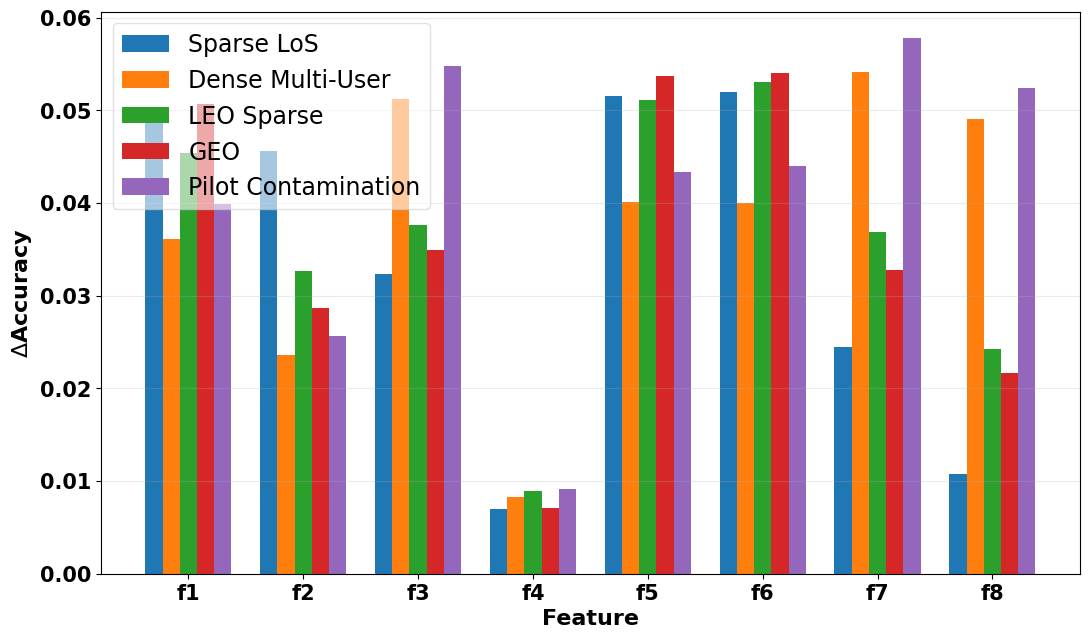}
		\label{fig:mil_acc}
	}
	
	\caption{Impact of each feature across scenarios: (a) AUC and (b) accuracy.}
\end{figure}

Figs.~\ref{fig:mil_auc} and \ref{fig:mil_acc} show environment-dependent feature importance via leave-one-feature-out AUC and accuracy degradation in NTN scenarios. In sparse LoS settings, removing \(f_{k,u,1}\), \(f_{k,u,5}\), or \(f_{k,u,6}\) leads to the largest degradation, highlighting the role of stable geometry and spatial structure. In dense or pilot-contaminated scenarios, \(f_{k,u,3}\), \(f_{k,u,7}\), and \(f_{k,u,8}\) become more influential due to post-combining separability and interference effects. The results also show that \(f_{k,u,6}\) dominates in sparse LEO links because of strong Doppler signatures, while GEO scenarios depend more on spatial-domain features.
The ablation results show no universally dominant feature. This motivates causal meta-learning over the full feature set.
Figs.~\ref{fig:auc_ave} and~\ref{fig:acc_ave} report the average \(\Delta \mathrm{AUC}_i\) and accuracy variations across NTN conditions. \(f_{k,u,4}\) is generally the least influential, while \(f_{k,u,6}\) is more relevant in LEO, low-SNR, and far-attacker scenarios. The interference-aware features \(f_{k,u,7}\) and \(f_{k,u,8}\) are most useful in dense, wide-angular-spread, and pilot-contaminated environments, whereas \(f_{k,u,1}\) and \(f_{k,u,5}\) remain strong in sparse LoS, GEO, and large-array settings.
Figs.~\ref{fig:heat_auc} and~\ref{fig:heat_acc} further confirm spatial dominance in sparse LoS and interference dominance in dense scenarios.

\begin{figure}[t]
	\centering
	\captionsetup[subfloat]{font=footnotesize}
	
	\subfloat[]{
		\includegraphics[width=0.46\columnwidth]{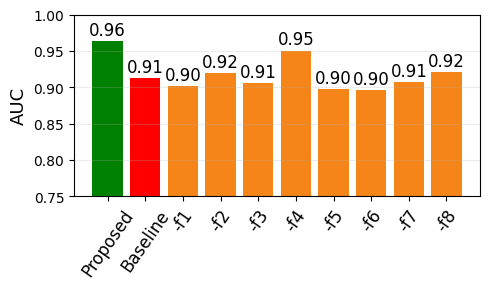}
		\label{fig:auc_ave}
	}
	\hfill
	\subfloat[]{
		\includegraphics[width=0.46\columnwidth]{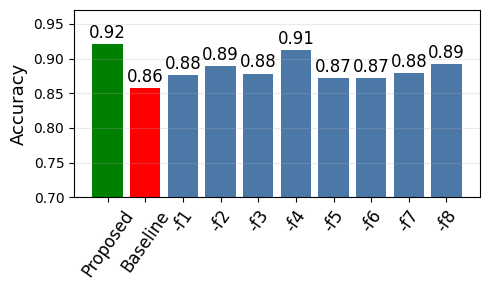}
		\label{fig:acc_ave}
	}
	
	\subfloat[]{
		\includegraphics[width=0.46\columnwidth]{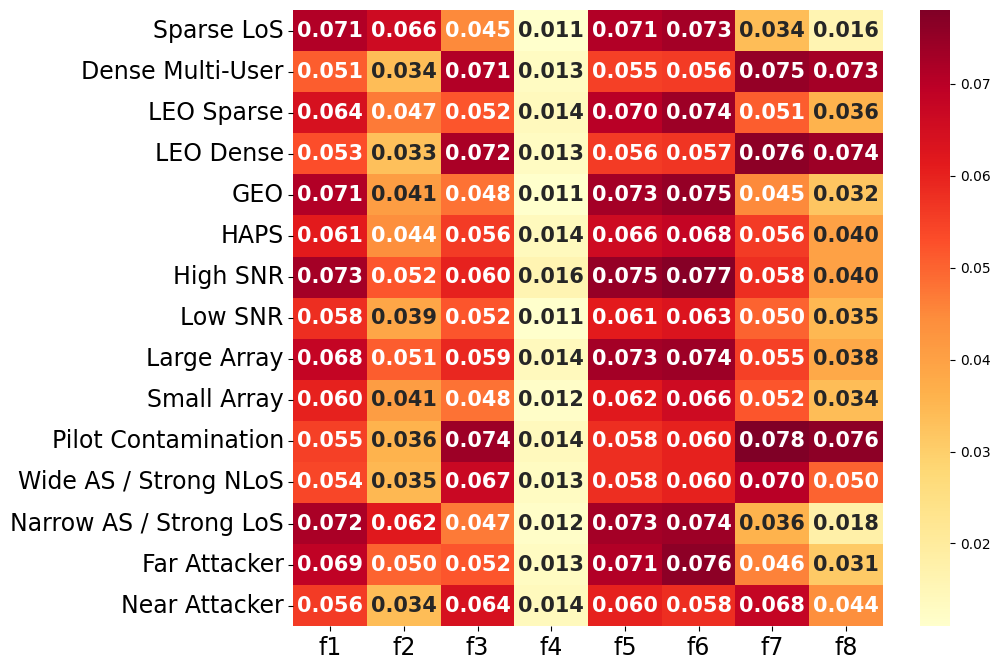}
		\label{fig:heat_auc}
	}
	\hfill
	\subfloat[]{
		\includegraphics[width=0.46\columnwidth]{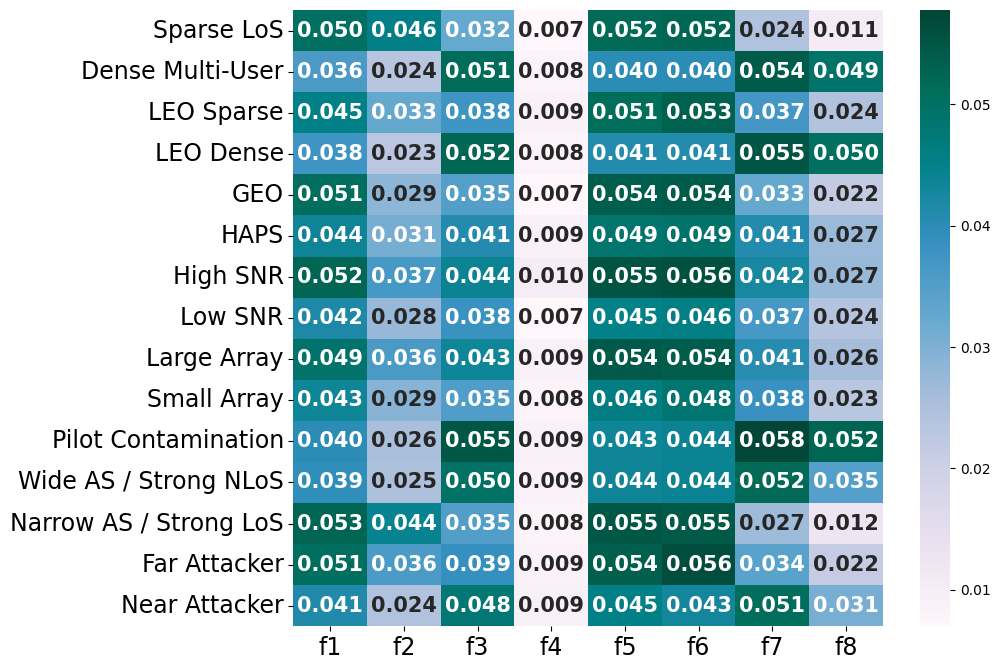}
		\label{fig:heat_acc}
	}
	
	\caption{AUC and accuracy analyses: (a),(b) feature-removal impacts on AUC and accuracy; (c),(d) feature-importance heatmaps for AUC and accuracy.}
	\label{fig:combined_four}
\end{figure}

\section{Conclusion}
\label{sec:conclusion}

This paper proposed SAFA-MZ, a causal meta-learning framework for DPLA in NTNs under pilot contamination. SAFA-MZ builds an adaptive and distributed physical-layer fingerprint using multi-node features and adjusts the fusion process to changing NTN environments. It combines causal analysis with meta-learning to improve robustness under distribution shifts and enable fast adaptation with few labeled samples. A two-stage authentication mechanism performs local authentication first and activates TDOA-based distributed verification only when needed to reduce overhead. Simulation results showed that SAFA-MZ achieves $92\%$ accuracy and $96\%$ AUC, outperforming baseline and single-feature methods.

\appendices

\section{Proofs}

\subsection{Backdoor identification of $\mu_{k,u}(t)$}
\label{app:back}

\begin{proposition}
	Under consistency, conditional exchangeability
	\(Y_{k,u}(t)\perp T\mid Z_{k,u}\), and positivity
	\(\mathbb{P}(T=t\mid Z_{k,u}=z)>0\), the interventional mean is identifiable as
	\[
	\mu_{k,u}(t)
	= \mathbb{E}[Y_{k,u}\mid \mathrm{do}(T=t)]
	= \mathbb{E}_{Z_{k,u}}\!\left[
	\mathbb{E}[Y_{k,u}\mid T=t,Z_{k,u}]
	\right].
	\]
\end{proposition}

\begin{proof}
	By iterated expectation and using conditional exchangeability, we have
	\begin{align}
		\mu_{k,u}(t) =& \mathbb{E}[Y_{k,u}(t)]
		= \mathbb{E}_{Z_{k,u}}\!\left[\mathbb{E}[Y_{k,u}(t)\mid Z_{k,u}]\right] \nonumber \\
		=& \mathbb{E}_{Z_{k,u}}\!\left[ \mathbb{E}[Y_{k,u}(t)\mid T=t, Z_{k,u}]\right]. \nonumber
	\end{align}
	By consistency, whenever $T=t$ we have $Y_{k,u}(t)=Y_{k,u}$, hence
	\(
	\mathbb{E}[Y_{k,u}(t)\mid T=t, Z_{k,u}]
	= \mathbb{E}[Y_{k,u}\mid T=t, Z_{k,u}].
	\)
	Combining the above equalities yields
	\[
	\mu_{k,u}(t)
	= \mathbb{E}_{Z_{k,u}}\!\left[\mathbb{E}[Y_{k,u}\mid T=t, Z_{k,u}]\right].
	\]
	Positivity guarantees that the conditional expectation on the right-hand side is well-defined on the support of $Z_{k,u}$.
\end{proof}

\subsection{Double robustness of the estimator in (48)}
\label{app:DR}

\begin{theorem}
	Define the population estimating function:
	\[
	\psi_{t}
	\triangleq
	m_t(Z_{k,u})
	+
	\frac{\mathbf{1}\{T=t\}}{e_t(Z_{k,u})}\Big(Y_{k,u}-m_t(Z_{k,u})\Big).
	\]
	Assume positivity and bounded moments of $Y_{k,u}$.
	Then
	\(\mathbb{E}[\psi_{t}] = \mu_{k,u}(t)\).
	Consequently, the sample estimator in (48) is consistent under standard regularity conditions (e.g., i.i.d. sampling and fixed or cross‑fitted nuisance estimators).
\end{theorem}

\begin{proof}
	We prove the result in the two standard cases.
	
	\paragraph{Case 1: the outcome model is correct.}
	\begin{align*}
		&\mathbb{E}[\psi_{t}\mid Z_{k,u}=z] \\ \nonumber
		&=
		m_t(z)
		+
		\frac{1}{e_t(z)}
		\mathbb{E}\!\left[\mathbf{1}\{T=t\}\big(Y_{k,u}-m_t(z)\big)\mid Z_{k,u}=z\right] \\ \nonumber
		&=
		m_t(z)
		+
		\frac{1}{e_t(z)} \mathbb{P}(T=t\mid Z_{k,u}=z) \nonumber \\
		&\quad\mathbb{E}\left[Y_{k,u}-m_t(z)\mid T=t,Z_{k,u}=z\right].
	\end{align*}
	The last conditional expectation is zero, and therefore
	\[
	\mathbb{E}[\psi_{t}\mid Z_{k,u}=z]=m_t(z)=\mathbb{E}[Y_{k,u}\mid T=t,Z_{k,u}=z].
	\]
	Taking expectation over $Z_{k,u}$ gives
	\[
	\mathbb{E}[\psi_{t}]
	=
	\mathbb{E}_{Z_{k,u}}\!\left[\mathbb{E}[Y_{k,u}\mid T=t,Z_{k,u}]\right]
	=
	\mu_{k,u}(t),
	\]
	where the last equality follows from the proof above.
	
	\paragraph{Case 2: the propensity model is correct.}
	\begin{align*}
		&\mathbb{E}[\psi_{t}\mid Z_{k,u}=z]\\
		&=
		m_t(z)
		+
		\frac{1}{e_t(z)} \mathbb{E}\!\left[\mathbf{1}\{T=t\}Y_{k,u}\mid Z_{k,u}=z\right] \\
		&\quad- \frac{m_t(z)}{e_t(z)}
		\mathbb{E}\!\left[\mathbf{1}\{T=t\}\mid Z_{k,u}=z\right] \\
		&=
		m_t(z)
		+
		\frac{\mathbb{P}(T=t\mid Z_{k,u}=z)\,\mathbb{E}[Y_{k,u}\mid T=t,Z_{k,u}=z]}{e_t(z)} \\
		&\quad- \frac{m_t(z)}{e_t(z)}\mathbb{P}(T=t\mid Z_{k,u}=z) \\
		&=
		\mathbb{E}[Y_{k,u}\mid T=t,Z_{k,u}=z].
	\end{align*}
	Taking expectation over $Z_{k,u}$ again yields $\mathbb{E}[\psi_{t}] = \mu_{k,u}(t)$.
	
	Thus, the estimator is DR, remaining consistent if either the propensity or outcome model is correctly specified.
\end{proof}

\subsection{Correct mediation decomposition}
\label{app:med_dec}

\begin{proof}
	Add and subtract the same intermediate term $\mathbb{E}[Y_{k,u}(t_1,\mathbf{f}_{k,u}(t_0))]$ to prove the decomposition:
	\begin{align*}
		\mathrm{ACE}_{k,u}(t_1,t_0)
		&=
		\mathbb{E}\!\left[\widehat{Y}_{k,u}(t_1,\mathbf{f}_{k,u,t_1}) - \widehat{Y}_{k,u}(t_0,\mathbf{f}_{k,u}(t_0))\right] \\
		&=
		\mathbb{E}\!\left[\widehat{Y}_{k,u}(t_1,\mathbf{f}_{k,u,t_1}) - \widehat{Y}_{k,u}(t_1,\mathbf{f}_{k,u}(t_0))\right] \\
		&\quad+ \mathbb{E}\!\left[\widehat{Y}_{k,u}(t_1,\mathbf{f}_{k,u}(t_0)) - \widehat{Y}_{k,u}(t_0,\mathbf{f}_{k,u}(t_0))\right] \\
		&=
		\mathrm{NIE}_{k,u}(t_1,t_0)+\mathrm{NDE}_{k,u}(t_1,t_0).
	\end{align*}
\end{proof}

\subsection{Zero IRM penalty under an invariant representation}
\label{app:zero_IRM}

\begin{proposition}
	Assume there exists a representation $\Phi_{\phi^\star}$ such that for every triple $(e,k,u)$, $w=1$ is a stationary point of the risk
	\(R_{e,k,u}(w\circ \Phi_{\phi^\star})\).
	Then $\mathcal{R}_{\mathrm{IRM}}(\phi^\star)=0$.
\end{proposition}

\begin{proof}
	Substituting \(\phi^\star\) into (52), if \(w=1\) is stationary for every \((e,k,u)\), then each gradient term is zero. Hence, all summands vanish and \(\mathcal{R}_{\mathrm{IRM}}(\phi^\star)=0\).
\end{proof}

\bibliographystyle{IEEEtran}
\bibliography{refs_2_causal}

\vfill

\end{document}